\documentclass[reqno,12pt]{amsart}

\usepackage{amsmath}
\usepackage{latexsym}
\usepackage{amssymb}
\usepackage{hyperref}
\usepackage{graphicx}
\usepackage{epstopdf}
\usepackage{ifpdf}
\ifpdf
\fi

\makeatletter
 
 \newcommand{\Rmnum}[1]{\expandafter\@slowromancap\romannumeral #1@}
 \makeatother

\newtheorem{theorem}{Theorem}[section]
\newtheorem{definition}{Definition}[section]
\newtheorem{proposition}[theorem]{Proposition}
\newtheorem{lemma}[theorem]{Lemma}

\newtheorem{remark}[theorem]{Remark}

\newcommand{\R}{{\mathbb R}}

\newcommand{\C}{{\mathbb C}}

\newcommand{\be}{\begin{equation}}
\newcommand{\ee}{\end{equation}}
\newcommand{\bea}{\begin{eqnarray}}
\newcommand{\eea}{\end{eqnarray}}
\newcommand{\ba}{\begin{array}}
\newcommand{\ea}{\end{array}}

\newcommand{\ol}{\overline}

\newcommand{\ra}{\rightarrow}

\newcommand{\id}{\mathbb{I}}

\newcommand{\re}{\mathrm{Re}}
\newcommand{\im}{\mathrm{Im}}

\newcommand{\eps}{\varepsilon}

\newcommand{\sig}{\sigma}
\newcommand{\Sig}{\Sigma}
\newcommand{\lam}{\lambda}

\newcommand{\gam}{\gamma}
\newcommand{\Gam}{\Gamma}
\newcommand{\om}{\omega}
\newcommand{\Om}{\Omega}

\newcommand{\dta}{\delta}
\newcommand{\Dta}{\Delta}

\newcommand{\tha}{\theta}

\numberwithin{equation}{section}

\begin{document}
\title[Long-time asymptotics for the FL equation]{Leading-order temporal asymptotics of the Fokas-Lenells Equation without solitons}

\author[J.Xu]{Jian Xu}
\address{School of Mathematical Sciences\\
Fudan University\\
Shanghai 200433\\
People's  Republic of China}
\email{11110180024@fudan.edu.cn}

\author[E.Fan]{Engui Fan*}
\address{School of Mathematical Sciences, Institute of Mathematics and Key Laboratory of Mathematics for Nonlinear Science\\
Fudan University\\
Shanghai 200433\\
People's  Republic of China}
\email{correspondence author:faneg@fudan.edu.cn}


\date{\today}

\begin{abstract}
We use the Deift-Zhou method to obtain, in the solitonless sector, the leading order asymptotic of the solution
to the Cauchy problem of the Fokas-Lenells equation as $t\ra+\infty$ on the full-line .

\end{abstract}

\maketitle

\section{Introduction}
The Fokas-Lenells equation (FL equation shortly) is a completely integrable nonlinear partial differential equation which has been
derived as an integrable generalization of the nonlinear Schr\"odinger equation (NLS equation) using bi-Hamiltonian methods \cite{f}.
In the context of nonlinear optics, the FL equation models the propagation of nonlinear light pulses in monomode optical fibers when
certain higher-order nonlinear effects are taken into account \cite{l}. The FL equation is related to the NLS equation in the same way
as the Camassa-Holm equation associated with the KdV equation. The soliton solutions of the FL equation have been constructed via the
Riemann-Hilbert method in \cite{lf}. And The initial-boundary value problem for the FL equation on the half-line was studied in \cite{lf2}.
A simple N-bright-soliton solution was given by Lenells \cite{l2} and the N-dark soliton solution was obtained by means of B\"acklund
transformation \cite{v}. And Matsuno get the bright and dark soliton solutions for the FL equation in \cite{ym} and \cite{ym2} by a direct method.
\par
In this paper, we use the Riemann-Hilbert problem showed in
\cite{lf} to get the long-time asymptotics behavior of the solution
of the FL equation (\ref{FLe}) by the nonlinear steepest descent
method or Deift-Zhou method. The nonlinear steepest descent method
is introduced by Deift and Zhou in \cite{dz} in 1993, the history of
the long-time asymptotics problem also canbe found in \cite{dz}, and
it is the first time to obtain the long-time asymptotics behavior of
the solution rigorously, for the MKdV equation. Then it becomes a
most power tool for the long-time asymptotics of the nonlinear
evolution equations in complete integrable system, for example, the
non-focusing NLS equation \cite{diz}, the Sine-Gordon equation
\cite{cvz}, the KdV equation \cite{gt}, the Cammasa-Holm equation
\cite{amkst}, and so on. Deift and his collaborators extend this
method to analyse the small-dispersion problem for the KdV equation
and the semiclassical problem of the focusing NLS equation. And this
is also a very usefull tool in the asymptotics problem in orthogonal
polynomials and large $n$ limit problem in random matrix theory.
\par
For several soliton-bearing equations, for example, KdV,
Landau-Lifshitz, and NLS, and the reduced Maxwell-Bloch system, it
is well known that the dominant $O(1)$ asymptotic $t\rightarrow
\infty$ effect of the continuous spectra on the multisoliton
solutions is a shift in phase and position of their constituent
solitons \cite{s}. The purpose of our studies is to derive an
explicit functional form for the next-to-leading-order
$O(t^{-\frac{1}{2}})$ term of the effect of this interaction for the
Fokas-Lenells equation. An asymptotic investigation of the solution
can be divided into two stages: (i) the investigation of the
continuum (solitonless) component of the solution \cite{mz}; and
(ii) the inclusion of the soliton component via the application of a
'dressing' procedure \cite{zs} to the continuum background. The
purpose of this paper is to carry out, systematically, stage (i) of
the abovementioned asymptotic paradigm (since this phase of the
asymptotic procedure is rather technical and long in itself, the
completed results for stage (ii) are the subject of a forthcoming
article ). The results obtained in this paper are formulated as
theorems \ref{mainresult}.
\par
The outline of this paper is as follows: In Section 2 we recall some classic definition of Riemann-Hilbert problem and then, we write down the Riemann-Hilbert problem of the Fokas-Lenells equation. In Section 3 we analyse the leading order asymptotics of the solution of the Fokas-Lenells equation as $t\rightarrow +\infty$ via the Deift-Zhou method.

\section{The Riemann-Hilbert problem for the Fokas-Lenells equation}

\subsection{What a Riemann-Hilbert problem is}
In this subsection, we first explain what a Riemann-Hilbert problem is
\begin{definition}
Let the contour $\Gam$ be the union of a finite number of smooth and oriented curves (orientation means that each arc of $\Gam$ has a positive side and a negative side: the positive (respectively, negative) side lies to the left (respectively, right) as one traverses the contour in the direction of the arrow) on the Riemann sphere $\bar \C$ (i.e. the complex plane with the point at infinity) such that $\bar \C\backslash \Gam$ has only a finite number of connected components. Let $V(k)$ be an $2\times 2$ matrix defined on the contour $\Gam$. The Riemann-Hilbert problem $(\Gam,V)$  is the problem of finding an $2\times 2$ matrix-valued function $M(k)$ that satisfies
\begin{enumerate}
\item $M(k)$ is analytic for $k\in \bar \C\backslash \Gam$ and extends continuously to the contour $\Gam$.
\item $M_+(k)=M_-(k)V(k),\quad k\in\Gam$.
\item $M(k)\rightarrow \id,\quad as \quad k\rightarrow \infty.$
\end{enumerate}
\end{definition}
\par
The Riemann-Hilbert problem can be solved as follows (see, \cite{bc}). Assume that $V(k)$ admits some factorization
\be\label{BCRHPfac}
V(k)=b_-^{-1}(k)b_+(k),
\ee
where
\be\label{BCRHPbpm}
b_+(k)=\om_+(k)-\id,\quad b_-(k)=\id-\om_-(k).
\ee
And define
\be\label{BCRHPom}
\om(k)=\om_+(k)+\om_-(k).
\ee
\par
Let
\be\label{BCRHPcauchy}
(C_{\pm}f)(k)=\int_{\Gam}\frac{f(\xi)}{\xi-k_{\pm}}\frac{d\xi}{2\pi i},\quad k\in\Gam,f\in L^2(\Gam),
\ee
denote the Cauchy operator on $\Gam$. As is well known, the operator $C_{\pm}$ are bounded from $L^2(\Gam)$ to $L^2(\Gam)$, and $C_+-C_-=\Rmnum{1}$, here $\Rmnum{1}$ denote the identify operator.
\par
Define
\be\label{BCRHPCom}
C_{\om}f=C_+(f\om_-)+C_-(f\om_+)
\ee
for $2\times 2$ matrix-valued functions $f$. Let $\mu$ be the solution of the basic inverse equation
\be\label{BCRHPinverseeq}
\mu=\id+C_{\om}\mu.
\ee
Then
\be\label{BCRHPsol}
M(k)=\id+\int_{\Gam}\frac{\mu(\xi)\om(\xi)}{\xi-k}\frac{d\xi}{2\pi i},\quad k\in \bar \C\backslash \Gam,
\ee
is the solution of the Riemann-Hilbert problem. (See \cite{dz},P.322).

\subsection{Riemann-Hilbert problem for FL equation}
The Fokas-Lenells equation is
\be\label{FLe1}
iu_t-\nu u_{tx}+\gamma u_{xx}+\sig |u|^2(u+i\nu u_x)=0,\quad \sig=\pm 1.
\ee
where $\nu$ and $\gamma$ are constants.
\par
If we replaced $u(x,t)$ by $u(-x,t)$, we can see the sign of $\nu$ is the same as the $\gamma$'s. Hence, we can assume that
$\alpha=\frac{\gamma}{\nu}>0$ and $\beta=\frac{1}{\nu}$. Then we change the variable as follows:
\[
u\rightarrow \beta\sqrt{\alpha}e^{i\beta x}u,\qquad \sig\rightarrow -\sig
\]
the equation (\ref{FLe1}) can be changed into the desired form:
\be\label{FLe}
u_{tx}+\alpha\beta^2u-2i\alpha\beta u_x-\alpha u_{xx}+\sig i\alpha\beta^2|u|^2u_x=0,\quad \sig=\pm 1.
\ee
This equation admits Lax pair
\be\label{FLelax}
\left\{
\ba{l}
\Phi_x+ik^2\sig_3\Phi=kU_x\Phi\\
\Phi_t+i\eta^2\sig_3\Phi=[\alpha kU_x+\frac{i\alpha\beta^2}{2}\sig_3(\frac{1}{k}U-U^2)]\Phi.
\ea
\right.
\ee
where $U=\left(\ba{cc}0&u\\v&0\ea\right)$, $\eta=\sqrt{\alpha}(k-\frac{\beta}{2k})$  with $v=\sig \bar u$. And in the following of the paper
we just consider $\sig=1$.
\par
According to the paper \cite{lf}, we can get the Riemann-Hilbert problem of the Fokas-Lenells equation (\ref{FLe}) as follows:
\be\label{RHP}
\left\{
\ba{l}
M_+(x,t,k)=M_-(x,t,k)J(x,t,k),\quad k\in \R\cup i\R,\\
M(x,t,k)\rightarrow \id,\qquad k\rightarrow \infty.
\ea
\right.
\ee
\begin{figure}[th]
\centering
\includegraphics{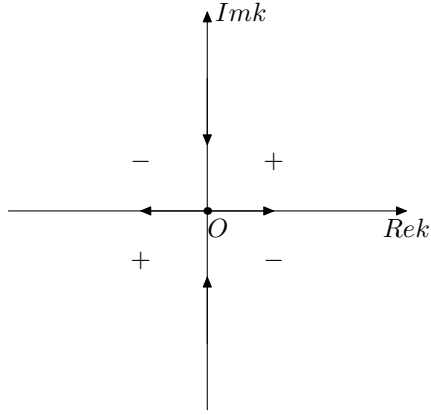}
\caption{The jump contour in the complex $k-$plane.}
\end{figure}
where the function $M(x,t,k)$ is defined by (4.24) in \cite{lf} and the jump matrix $J(x,t,k)$ is defined by
\be\label{Jump}
J(x,t,k)=e^{(-ik^2x-i\eta^2t)\hat \sig_3}\left(\ba{cc}\frac{1}{a(k)\ol{a(\bar k)}}&\frac{b(k)}{\ol{a(\bar k)}}\\
-\frac{\ol{b(\bar k)}}{a(k)}&1\ea\right)
\ee
with $\eta=\sqrt{\alpha}(k-\frac{\beta}{2k})$ and $a(k),b(k)$ are defined by (4.26) in \cite{lf}.
And $e^{\hat \sig_3}A=e^{\sig_3}Ae^{-\sig_3}$, here $A$ is a $2\times 2$ matrix.
We can also know that $a(k)\ol{a(\bar k)}-b(k)\ol{b(\bar k)}=1$ and $a(k)=\ol{a(\bar k)}$, $b(k)=-\ol{b(\bar k)}$ from \cite{lf}.
\par
We introduce $r(k)=\frac{\ol{b(\bar k)}}{a(k)}$, then the jump matrix $J(x,t,k)$ can be transformed into the following form:
\be\label{useJump}
J(x,t,k)=e^{(-ik^2x-i\eta^2t)\hat \sig_3}\left(\ba{cc}1-r(k)\ol{r(\bar k)}&\ol{r(\bar k)}\\
-r(k)&1\ea\right)
\ee
\par
The solution of Fokas-Lenells equation (\ref{FLe}) can be expressed by
\be\label{solfromRHP}
u_x(x,t)=2im(x,t)e^{4i\int_{-\infty}^{x}|m|^2(x',t)dx'}
\ee
where
\be\label{m}
m(x,t)=\lim_{k\rightarrow\infty}(kM(x,t,k))_{12}
\ee
with $M(x,t,k)$ is the unique solution of the Riemann-Hilbert problem (\ref{RHP}).

\begin{remark}
In this paper, we consider the case when $a(k)$ has no zeros, that is without solitons, the unique solvability of the Riemann-Hilbert problem in (\ref{RHP}) is a consequence of a vanishing lemma 4.2 in \cite{lf}.
\end{remark}

\section{The Long-time asymptotics for the Fokas-Lenells equation}
In this section, we get the asymptotics behavior of the solution of the Fokas-Lenells equation (\ref{FLe}) as $t\rightarrow \infty$ by
the Deift-Zhou method \cite{dz}.
\par
Let $F(x,t,k)=k^2x+\eta^2t$ and $\tha(k)=k^2\frac{x}{t}+\eta^2$, then $F=t\tha$.
\subsection{Case 1: $\frac{x}{t}+\alpha<0$}

In this case, the real part of $i\tha(k)$ has the signature
\be\label{case1Reitha}
\re{i\tha(k)}\left\{
\ba{l}
>0,\quad if \quad \im{k^2}>0,\\
<0,\quad if \quad \im k^2<0.
\ea
\right.
\ee
showed in Figure 2.
\begin{figure}[th]
\centering
\includegraphics{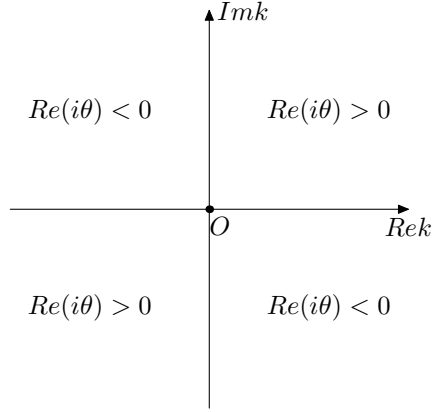}
\caption{The signature table of $\re(i\tha)$ in the case 1.}
\end{figure}
\par
The jump matrix $J(x,t,k)$ ,i.e. (\ref{useJump}), has an factorization
\be\label{jumpfaccase1}
J(x,t,k)=
\left(\ba{cc}1&0\\\frac{-r(k)}{1-r(k)\ol{r(\bar k)}}e^{2it\tha(k)}&1\ea\right)\left(\ba{cc}1-r(k)\ol{r(\bar k)}&0\\0&\frac{1}{1-r(k)\ol{r(\bar k)}}\ea\right)\left(\ba{cc}1&\frac{\ol{r(\bar k)}}{1-r(k)\ol{r(\bar k)}}e^{-2it\tha(k)}\\0&1\ea\right),
\ee
We find that the transformation
\be\label{transformcase1}
\tilde M(x,t,k)=M(x,t,k)\tilde \dta^{-\sig_3},
\ee
leads to the Riemann-Hilbert problem
\be\label{transRHPcase1}
\left\{
\ba{l}
\tilde M_+(x,t,k)=\tilde M_-(x,t,k)\tilde J(x,t,k),\qquad \im k^2=0,\\
\tilde M\rightarrow \id,\qquad k\rightarrow \infty.
\ea
\right.
\ee
with jump matrix $\tilde J(x,t,k)$ that admits the lower/upper factorization
\be\label{transjumpcase1}
\tilde J(x,t,k)=\left(\ba{cc}1&0\\\frac{-r(k)}{1-r(k)\ol{r(\bar k)}}\frac{1}{\tilde \dta^2(k)_-}e^{2it\tha(k)}&1\ea\right)
\left(\ba{cc}1&\frac{\ol{r(\bar k)}}{1-r(k)\ol{r(\bar k)}}\tilde \dta^2(k)_+e^{-2it\tha(k)}\\0&1\ea\right)=\tilde J_1^{-1}\tilde J_2
\ee
if the function $\tilde \dta(k)$ solves the scalar Riemann-Hilbert problem
\be\label{case1tdtaRHP}
\left\{
\ba{ll}
\tilde \dta_+(k)=\tilde \dta_-(k)(1-r(k)\ol{r(\bar k)}),&\im k^2=0,\\
\tilde \dta(k)\rightarrow 1,&k\rightarrow \infty.
\ea
\right.
\ee
The solution for the Riemann-Hilbert problem for $\tilde \dta$ has the explicit form
\be\label{case1tdtasol}
\tilde \dta(k)=e^{\frac{1}{2\pi i}\int_{\R\cup i\R}\frac{\log{(1-r(k')\ol{r(\bar k')})}}{k'-k}dk'}.
\ee
\par
Without loss of generality, we may assume that the left factor of (\ref{transjumpcase1}) extends analytically to the region $\im k^2<0$ and continuous in the closure of the region. Then the right factor extends the region $\im k^2>0$.
\par
Our Riemann-Hilbert problem on $\R\cup i\R$ is equivalent to a new Riemann-Hilbert problem on the contour
\be\label{case1newcontour}
\tilde \Sig=e^{i\frac{\pi}{6}}\R\cup e^{-i\frac{\pi}{6}}\R\cup e^{i\frac{\pi}{3}}\R\cup e^{-i\frac{\pi}{3}}\R,
\ee
where the orientation of the contour $\tilde \Sig$ and the new function $\hat M(x,t,k)$ are given in the following
\be\label{case1hatM}
\hat M(x,t,k)=\left\{
\ba{ll}
\tilde M(x,t,k),&k\in \hat D_2\cup\hat D_5\cup \hat D_8\cup\hat D_{11},\\
\tilde M(x,t,k)\tilde J_2^{-1},&k\in \hat D_1\cup\hat D_3\cup\hat D_7\cup \hat D_9,\\
\tilde M(x,t,k)\tilde J_1^{-1},&k\in \hat D_4\cup\hat D_6\cup\hat D_{10}\cup\hat D_{12}.
\ea
\right.
\ee
where the domains $\{\hat D_j\}_1^{12}$ are showed in Figure 3.
\begin{figure}[th]
\centering
\includegraphics{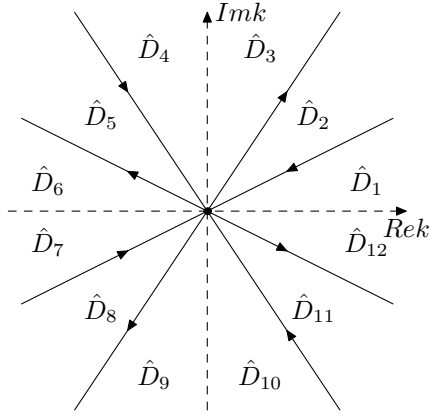}
\caption{The contour $\tilde \Sig$ and regions in case 1.}
\end{figure}
Then one can verify
\be\label{case1hatMRHP}
\left\{
\ba{ll}
\mbox{$\hat M$ is analytic off $\tilde \Sig$ ($\hat M$ is analytic across $\R\cup i\R$),}&\\
\hat M_+(x,t,k)=\hat M_-(x,t,k)\hat J(x,t,k),&k\in\tilde \Sig,\\
\hat M(x,t,k)\rightarrow \id,&k\rightarrow\infty.
\ea
\right.
\ee
where
\be\label{casehatjump}
\hat J(x,t,k)=\left\{
                  \ba{ll}
                    \left(\ba{cc}1&\frac{\ol{r(\bar k)}}{1-r(k)\ol{r(\bar k)}}\tilde \dta^2(k)_+e^{-2it\tha(k)}\\0&1\ea\right)^{-1},&k\in \tilde \Sig\cap\{\im k^2>0\},\\
                    \left(\ba{cc}1&0\\\frac{-r(k)}{1-r(k)\ol{r(\bar k)}}\frac{1}{\tilde \dta^2(k)_-}e^{2it\tha(k)}&1\ea\right),&k\in \tilde \Sig\cap\{\im k^2<0\}.
                   \ea
                \right.
\ee
\begin{theorem}
As $t\rightarrow \infty$,
\be\label{case1asy}
||\hat M_{\pm}(x,t,k)-\id||_{L^2(\tilde \Sig)}\rightarrow 0,\quad rapidly.
\ee
$u_x(x,t)$ and therefore $u$ decay rapidly as $t\rightarrow\infty$.
\end{theorem}
\begin{proof}
Since the Riemann-Hilbert problem for $M$ and the Riemann-Hilbert problem for $\hat M$ are equivalent, the existence of the solution of $M$ implies the existence
of the solution of $\hat M$.
\par
We make the trivial factorization
\[
\hat J(x,t,k)=b_-^{-1}b_+,\quad b_-=\id,b_+=\hat J.
\]
and define $\hat \om$ as (\ref{BCRHPom}).
Then as section 2 (also see, \cite{bc} or \cite{dz}) , we obtain the solution of the Riemann-Hilbert problem for $\hat M$,
\be\label{case1RHPsol}
\hat M(x,t,k)=\id+\int_{\hat \Sig}\frac{\hat \mu(x,t,\xi)\hat \om(x,t,\xi)}{\xi-k}\frac{d\xi}{2\pi i},\quad k\in \C\backslash\hat \Sig.
\ee
where $\hat \mu$ is the solution of the singular integral equation $\hat \mu=\id +C_{\hat \om}\hat \mu$, where $C_{\hat \om}$ defined as (\ref{BCRHPCom}) with $\om$ replaced by $\hat \om$.
Since
\[
||\hat J(x,t,k)-\id||_{L^2(\hat\Sig)\cap L^{\infty}(\hat \Sig)}\rightarrow 0,\quad exponentially,\quad as\quad t\rightarrow \infty,
\]
by (\ref{case1RHPsol}),
\[
||\hat M-\id||_{L^2(\hat \Sig)}\rightarrow 0,\quad rapidly,\quad as\quad t\rightarrow \infty.
\]
Then, by (\ref{solfromRHP}) we get $u_x$ decays rapidly , and then $u$ decays rapidly, as $t\rightarrow \infty$.
\end{proof}
\subsection{Case 2: $\frac{x}{t}+\alpha>0$}
In this case, the real part of $i\tha(k)$ has the signature as the Figure 3. And we set $k_0=(\frac{\alpha \beta^2}{4(\frac{x}{t}+\alpha)})^{\frac{1}{4}}$.
\begin{figure}[th]
\centering
\includegraphics{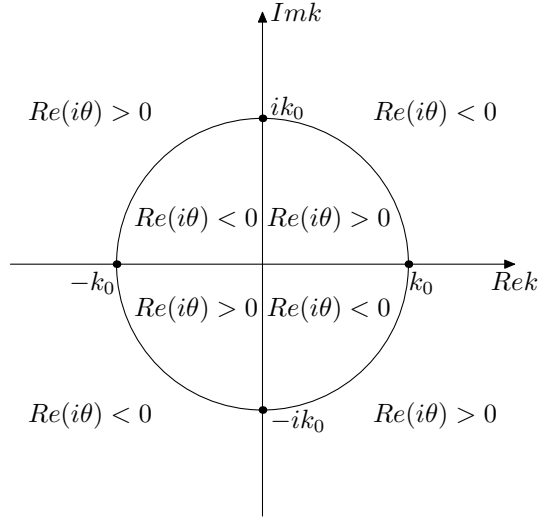}
\caption{The signature table of $\re(i\tha)$ in the case 2.}
\end{figure}
\par
The jump matrix $J(x,t,k)$ has the following factorization
\be\label{jumpfac}
J(x,t,k)=\left\{
\ba{l}
\left(\ba{cc}1&\ol{r(\bar k)}e^{-2it\tha(k)}\\0&1\ea\right)\left(\ba{cc}1&0\\-r(k)e^{2it\tha(k)}&1\ea\right),\\
\left(\ba{cc}1&0\\\frac{-r(k)}{1-r(k)\ol{r(\bar k)}}e^{2it\tha(k)}&1\ea\right)\left(\ba{cc}1-r(k)\ol{r(\bar k)}&0\\0&\frac{1}{1-r(k)\ol{r(\bar k)}}\ea\right)\left(\ba{cc}1&\frac{\ol{r(\bar k)}}{1-r(k)\ol{r(\bar k)}}e^{-2it\tha(k)}\\0&1\ea\right).
\ea
\right.
\ee
\subsubsection{The conjugate transform}

Introducing a scalar function $\dta(k)$ which solves the Riemann-Hilbert problem
\be\label{scalRHP}
\left\{
\ba{rll}
\dta(k)_+&=\dta(k)_-(1-r(k)\ol{r(\bar k)})&k\in\Sig=(-k_0,k_0)\cup i(-k_0,k_0),\\
&=\dta(k)_-=\dta(k)&k\in \{\im k^2=0\}\backslash \Sig.\\
&\dta(k)\rightarrow 1&k\rightarrow\infty.
\ea
\right.
\ee
The solution of this Riemann-Hilbert problem is given by
\be\label{dtadef}
\dta(k)=\left((\frac{k-k_0}{k})(\frac{k+k_0}{k})\right)^{i\vartheta}e^{\chi_+(k)}e^{\chi_-(k)}
\left((\frac{k}{k-ik_0})(\frac{k}{k+ik_0})\right)^{i\tilde\vartheta}e^{\tilde \chi_+(k)}e^{\tilde \chi_-(k)},
\ee
where
\begin{subequations}
\be\label{vartheta}
\vartheta=-\frac{1}{2\pi}\ln{(1-|r(k_0)|^2)},
\ee
\be\label{tildevartheta}
\tilde\vartheta=-\frac{1}{2\pi}\ln{(1+|r(ik_0)|^2)},
\ee
\be\label{chipm}
\chi_{\pm}(k)=\frac{1}{2\pi i}\int_{0}^{\pm k_0}\ln{\left(\frac{1-|r(k')|^2}{1-|r(k_0)|^2}\right)}\frac{dk'}{k'-k},
\ee
\be\label{tildechipm}
\tilde\chi_{\pm}(k)=\frac{1}{2\pi i}\int_{\pm ik_0}^{i0}\ln{\left(\frac{1-r(k')\ol{r(\bar k')}}{1+|r(ik_0)|^2}\right)}\frac{dk'}{k'-k}.
\ee
\end{subequations}
Moreover, for all $k\in \C$, $|\dta|$ and $|\dta^{-1}|$ are bounded.
\par
The conjugate transform is that
\be\label{firsttrans}
M^{(1)}(x,t,k)=M(x,t,k)\dta(k)^{-\sig_3}.
\ee
\begin{figure}[th]
\centering
\includegraphics{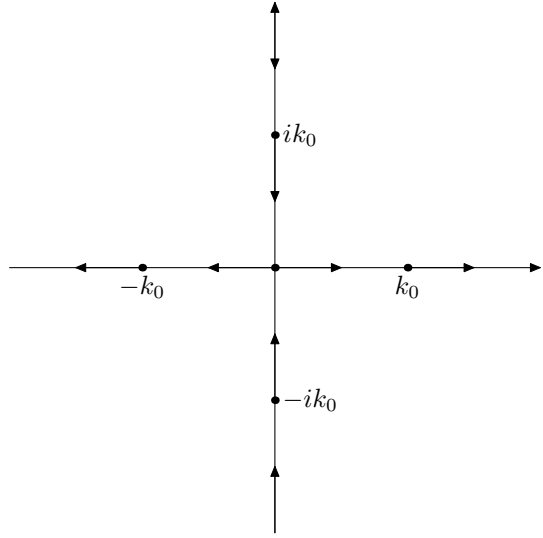}
\caption{The jump contour $\Sig^{(1)}$ for $M^{(1)}(x,t,k)$.}
\end{figure}
Then we can get the Riemann-Hilbert problem of $M^{(1)}(x,t,k)$
\be\label{M1RHP}
\left\{
\ba{ll}
M^{(1)}(x,t,k)_+=M^{(1)}(x,t,k)_-J^{(1)}(x,t,k),&k\in \R\cup i\R\\
M^{(1)}(x,t,k)\rightarrow \id,&k\rightarrow\infty.
\ea
\right.
\ee
where
\be\label{J1def}
J^{(1)}(x,t,k)=\left\{
\ba{l}
\left(\ba{cc}1&0\\\frac{-r(k)}{1-r(k)\ol{r(\bar k)}}\frac{1}{\dta^2(k)_-}e^{2it\tha(k)}&1\ea\right)\left(\ba{cc}1&\frac{\ol{r(\bar k)}}{1-r(k)\ol{r(\bar k)}}\dta^2(k)_+e^{-2it\tha(k)}\\0&1\ea\right),k\in\Sig^{(1)}=\Sig,\\
\left(\ba{cc}1&\ol{r(\bar k)}\dta^2(k)e^{-2it\tha(k)}\\0&1\ea\right)\left(\ba{cc}1&0\\-r(k)\frac{1}{\dta^2(k)}e^{2it\tha(k)}&1\ea\right),\quad k\in\{\im{k^2=0}\}\backslash\Sig^{(1)}.
\ea
\right.
\ee
\begin{figure}[th]\label{fig6}
\centering
\includegraphics{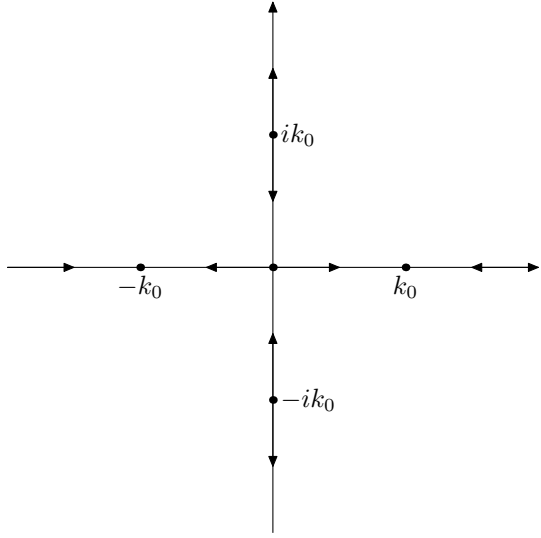}
\caption{The jump contour $\tilde \Sig^{(1)}$ for $\tilde M^{(1)}(x,t,k)$.}
\end{figure}
Then we reverse the direction of the part of $\{\im{k^2=0}\}\backslash\Sig^{(1)}$, we have
\be\label{M1RHPtilde}
\left\{
\ba{ll}
M^{(1)}(x,t,k)_+=M^{(1)}(x,t,k)_-\tilde J^{(1)}(x,t,k),&k\in \R\cup i\R\\
M^{(1)}(x,t,k)\rightarrow \id,&k\rightarrow\infty.
\ea
\right.
\ee
\small{
\be\label{J1defrealuse}
\tilde J^{(1)}(x,t,k)=\left\{
\ba{l}
\left(\ba{cc}1&0\\\frac{-r(k)}{1-r(k)\ol{r(\bar k)}}\frac{1}{\dta^2(k)_-}e^{2it\tha(k)}&1\ea\right)\left(\ba{cc}1&\frac{\ol{r(\bar k)}}{1-r(k)\ol{r(\bar k)}}\dta^2(k)_+e^{-2it\tha(k)}\\0&1\ea\right),k\in\tilde \Sig^{(1)}=\Sig\\
\left(\ba{cc}1&0\\r(k)\frac{1}{\dta^2(k)}e^{2it\tha(k)}&1\ea\right)\left(\ba{cc}1&-\ol{r(\bar k)}\dta^2(k)e^{-2it\tha(k)}\\0&1\ea\right),\quad k\in\{\im{k^2=0}\}\backslash\tilde \Sig^{(1)}.
\ea
\right.
\ee
}
\subsubsection{The second transform}
The main purpose of this section is to reformulate the original Riemann-Hilbert problem (\ref{M1RHPtilde}) as an equivalent Riemann-Hilbert problem on the augmented contour $\Sig^{(2)}$ (see Figure 7),
\be\label{Sig2}
\Sig^{(2)}=L\cup  L_0\cup\bar L\cup \bar L_0\cup\R\cup i\R.
\ee
where $L=L_1\cup\tilde L_1\cup L_2\cup\tilde L_2$,
\par
Denote the contour
\be\label{case2L}
\ba{ll}
L_1=\{k=k_0+uk_0e^{i\frac{3\pi}{4}},\quad u\in (-\infty,\frac{1}{\sqrt{2}}]\},
&\tilde L_1=\{k=ik_0+uk_0e^{-i\frac{\pi}{4}},\quad u\in (-\infty,\frac{1}{\sqrt{2}}]\}\\
L_2=\{k=-k_0+uk_0e^{-i\frac{\pi}{4}},\quad u\in (-\infty,\frac{1}{\sqrt{2}}]\},
&\tilde L_2=\{k=ik_0+uk_0e^{i\frac{\pi}{4}},\quad u\in (-\infty,\frac{1}{\sqrt{2}}]\}
\ea
\ee
Denote the contour
\be\label{case2L0}
L_0=\{uk_0e^{i\frac{\pi}{4}},\quad u\in [-\frac{1}{\sqrt{2}},\frac{1}{\sqrt{2}}]\}.
\ee
Denote the contour
\be\label{case2Leps}
\ba{rrl}
L_{\eps}&=&L_{1\eps}\cup\tilde L_{1\eps}\cup L_{2\eps}\cup \tilde L_{2\eps}\\
&=&\{k=k_0+uk_0e^{i\frac{3\pi}{4}},\quad \eps<u\le\frac{1}{\sqrt{2}}\}\cup \{k=ik_0+uk_0e^{-i\frac{\pi}{4}},\quad u\in (\eps,\frac{1}{\sqrt{2}}]\}\\
&&\cup\{k=-k_0+uk_0e^{-i\frac{\pi}{4}},\quad u\in (\eps,\frac{1}{\sqrt{2}}]\}\cup \{k=ik_0+uk_0e^{i\frac{\pi}{4}},\quad u\in (\eps,\frac{1}{\sqrt{2}}]\}
\ea
\ee
\par
Following the method in \cite{dz}, we can have
\begin{proposition}
Let
\be\label{rhodef}
\rho(k)=\left\{
\ba{ll}
\rho_1(k)=\frac{\ol{r(\bar k)}}{1-r(k)\ol{r(\bar k)}},&k\in \Sig\\
\rho_2(k)=-\ol{r(\bar k)},&k\in\{\im{k^2=0}\}\backslash\Sig.
\ea
\right.
\ee
Then $\rho$ has a decomposition
\be\label{rhoanalic}
\rho(k)=h_{\Rmnum{1}}(k)+(h_{\Rmnum{2}}(k)+R(k)),
\ee
where $h_{\Rmnum{1}}(k)$ is small and $h_{\Rmnum{2}}(k)$ has an analytic continuation to $L$ and $L_0$. For example, if $\rho(k)=r(k)$ as $k>k_0$, $h_{\Rmnum{2}}(k)$ of this function $\rho(k)$ has an analytic continuation to the first quadrant. And $R(k)$ is piecewise rational ($R(k)=0$, if $k\in L_0$) function.
\par
And $R,h_{\Rmnum{1}}, h_{\Rmnum{2}}$ satisfy
\begin{subequations}
\be\label{case2h1}
|e^{-2it\tha(k)}h_{\Rmnum{1}}(k)|\le\frac{c}{(1+|k|^2)t^l}, for \quad z\in \R\cup i\R,
\ee

\be\label{case2h2l}
|e^{-2it\tha(k)}h_{\Rmnum{2}}(k)|\le \frac{c}{(1+|k|^2)t^l},\quad k\in L,\quad k_0<M.
\ee
\be\label{case2h2l0}
|e^{-2it\tha(k)}h_{\Rmnum{2}}(k)| \le ce^{-t\frac{\alpha\beta^2}{4k_0^2}},\quad k\in L_0,\quad k_0<M.
\ee
and
\be\label{case2R}
|e^{-2it\tha(k)}R(k)|\le ce^{-\frac{\eps^2\alpha\beta^2}{M^2}t},\quad k\in L_{\eps}.
\ee
\end{subequations}
for arbitrary natural number $l$, for sufficiently large constants $c$, for some fixed positive constant $M$.
\end{proposition}
\begin{proof}
See appendix.
\end{proof}
\begin{remark}
Taking conjugate $\ol{\rho(k)}=\ol{h_{\Rmnum{1}}(k)}+\ol{h_{\Rmnum{2}}(\bar k)}+\ol{R(\bar k)}$ leads to the same estimates for $e^{2it\tha(k)}\ol{h_{\Rmnum{1}}(k)},e^{2it\tha(k)}\ol{h_{\Rmnum{2}}(\bar k)}$ and $e^{2it\tha(k)}\ol{R(\bar k)}$ on $\R\cup i \R\cup \bar L\cup \bar L_0$.
\end{remark}
\par
From the Riemann-Hilbert problem (\ref{M1RHPtilde}) and formula (\ref{J1defrealuse}), the Riemann-Hilbert problem across $\R\cup i\R$ oriented as Figure 6 is given by
\be\label{M1RHPBCtype}
\left\{
\ba{ll}
M^{(1)}(x,t,k)_+=M^{(1)}(x,t,k)_-(b_-)^{-1}b_+,&k\in \R\cup i\R\\
M^{(1)}(x,t,k)\rightarrow \id,&k\rightarrow\infty.
\ea
\right.
\ee
where
\be\label{bp}
b_{+}=\id+\om_+=\dta_{+}^{\hat \sig_3}e^{-it\tha(k)\hat \sig_3}\left(\ba{cc}1&\rho(k) \\ 0&1\ea\right),
\ee
\be\label{bm}
b_{-}=\id-\om_-=\dta_{-}^{\hat \sig_3}e^{-it\tha(k)\hat \sig_3}\left(\ba{cc}1&0 \\ \ol{\rho(\bar k)}&1\ea\right),
\ee
and $\rho$ is given by (\ref{rhodef}).
\par
We write
\begin{subequations}
\be\label{bpfac}
b_+=b^o_+b^a_+=(\id+\om^o_+)(\id+\om^a_+)=\left(\ba{cc}1&h_{\Rmnum{1}}(k)\dta_+^2e^{-2it\tha}\\0&1\ea\right)\left(\ba{cc}1&(h_{\Rmnum{2}}(k)+R(k))\dta_+^2e^{-2it\tha}\\0&1\ea\right),
\ee
\be\label{bmfac}
b_-=b^o_-b^a_-=(\id-\om^o_-)(\id-\om^a_-)=\left(\ba{cc}1&0\\\ol{h_{\Rmnum{1}}(\bar k)}\frac{1}{\dta_-^2}e^{2it\tha}&1\ea\right)\left(\ba{cc}1&0\\(\ol{h_{\Rmnum{2}}(\bar k)+R(\bar k)})\frac{1}{\dta_-^2}e^{2it\tha}&1\ea\right),
\ee
\end{subequations}
Now we can use the signature table of $\re{i\tha}$ showed in Figure 3 to open the jump contour for the Riemann-Hilbert problem of $M^{(1)}$ to the contours in Figure 7.
\begin{figure}[th]
\centering
\includegraphics{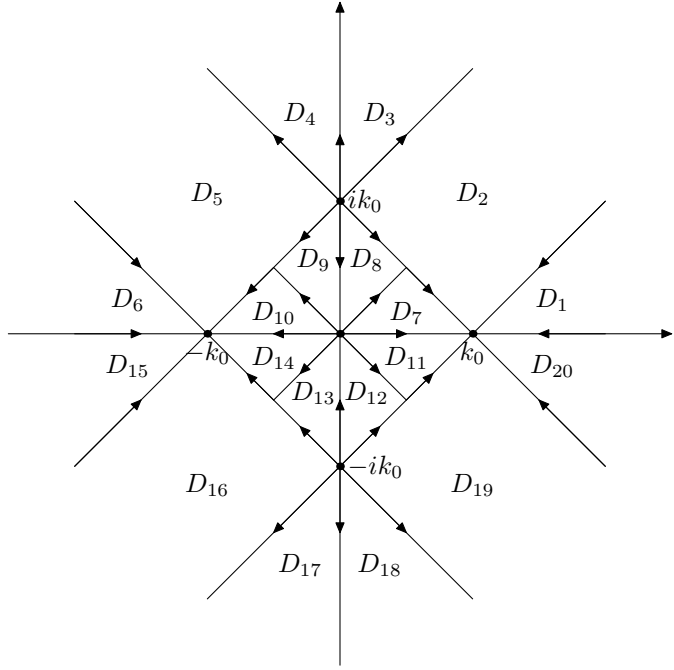}
\caption{The different regions $\{D_j\}_1^{20}$ of the complex $k-$plane.}
\end{figure}
Introducing $M^{(2)}(x,t,k)=M^{(1)}(x,t,k)\phi$, where $\phi$ is defined as follows:
\be\label{phidef}
\phi=\left\{
\ba{ll}
\id,&k\in D_2,D_5,D_{16},D_{19}\\
(b^a_-)^{-1},&k\in D_1,D_3,D_{15},D_{17},D_9,D_{10},D_{11},D_{12}\\
(b^a_+)^{-1},&k\in D_4,D_6,D_{18},D_{20},D_7,D_8,D_{13},D_{14}
\ea
\right.
\ee
where the regions $\{D_j\}_1^{20}$ are showed in Figure 7.
\par
Then the Riemann-Hilbert problem of $M^{(2)}(x,t,k)$ is defined
\be\label{M2RHP}
M^{(2)}_+(x,t,k)=M_-^{(2)}(x,t,k)J^{(2)}(x,t,k)
\ee
with
\be\label{J2def}
J^{(2)}(x,t,k)=\left\{
\ba{ll}
(b^o_-)^{-1}(b^o_+),&k\in R\cup i\R\\
\id^{-1}(b^a_+),&k\in L\cup {\bf L_0}\\
(b^a_-)^{-1}\id,&k\in\bar L\cup {\bf \bar L_0}
\ea
\right.
\ee

\par
Using the symbol $J^{(2)}(x,t,k)=b^{-1}_-(x,t,k)b_+(x,t,k)$, and set $\om_{\pm}(x,t,k)=\pm(b_{\pm}(x,t,k)-\id)$ , $\om(x,t,k)=\om_+(x,t,k)+\om_-(x,t,k)$. From section 2, we have
\be\label{M2sol}
M^{(2)}(x,t,k)=\id+\int_{\Sig^{(2)}}\frac{\mu(x,t,\xi)\om(x,t,\xi)}{\xi-k}\frac{d\xi}{2\pi i},\quad k\in\C\backslash\Sig^{(2)}.
\ee
And substituting (\ref{M2sol}) into (\ref{m}), we learn that
\be\label{m2}
\ba{rl}
m(x,t)=&\frac{1}{2}\lim_{k\rightarrow \infty}(k[\sig_3,M^{(2)}(x,t,k)])_{12},\\
=&-\frac{1}{2}([\sig_3,\int_{\Sig^{(2)}}\mu(x,t,\xi)\om(x,t,\xi)]\frac{d\xi}{2\pi i})_{12},\\
=&-\frac{1}{2}([\sig_3,\int_{\Sig^{(2)}}((\id-C_{\om})^{-1}\id)(\xi)\om(x,t,\xi)]\frac{d\xi}{2\pi i})_{12}.
\ea
\ee

\subsubsection{Transform to the Riemann-Hilbert problem of $M^{(3)}(x,t,k)$}
\par
{Follow the method of \cite{dz} P.323-329, we can reduce the Riemann-Hilbert problem of $M^{(2)}(x,t,k)$ to the Riemann-Hilbert problem of $M^{(3)}(x,t,k)$.}
\begin{figure}[th]
\centering
\includegraphics{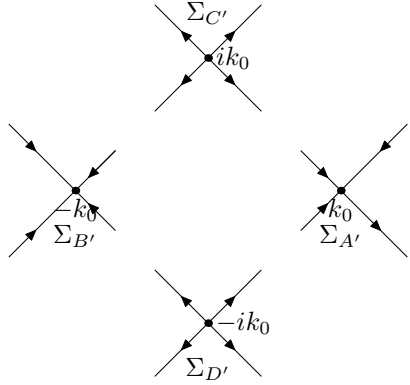}
\caption{The jump contour $\Sig^{(3)}$ for $M^{(3)}(x,t,k)$.}
\end{figure}
\par
Let $\om^e$ be a sum of three terms
\be
\om^e=\om^a+\om^b+\om^c+\om^d.
\ee
We then have the following:
\be
\ba{l}
\om^a=\om \mbox{ is supported on the $\R\cup i\R$ and consists of terms of type $h_{\Rmnum{1}}(k)$ and $\ol{ h_{\Rmnum{1}}(k)}$}.\\
\om^b=\om \mbox{ is supported on the $L\cup \bar L$ and consists of terms of type $h_{\Rmnum{2}}(k)$ and $\ol{ h_{\Rmnum{2}}(\bar k)}$}.\\
\om^c=\om \mbox{ is supported on the $L_{\eps}\cup \bar L_{\eps}$ and consists of terms of type $R(k)$ and $\ol{R(\bar k)}$}.\\
\om^d=\om \mbox{ is supported on the $L_0\cup \bar L_0$ }.
\ea
\ee
Set $\om'=\om-\om^e$. Then, $\om'=0$ on $\Sig^{(2)}\backslash \Sig^{(3)}$.
Thus, $\om'$ is supported on $\Sig^{(3)}$ with contribution to $\om$ from rational
terms $R$ and $\bar R$.
\begin{proposition}\label{omfanshu}
For $0<k_0<M$, we have
\begin{subequations}
\be\label{oma}
||\om^a||_{L^1(R\cup i\R)\cap L^2(\R\cup i\R)\cap L^{\infty}(\R\cup i\R)}\le \frac{c}{t^l},
\ee
\be\label{omb}
||\om^b||_{L^1(L\cup \bar L)\cap L^2(L\cup \bar L)\cap L^{\infty}(L\cup \bar L)}\le \frac{c}{t^l},
\ee
\be\label{omc}
||\om^c||_{L^1(L_{\eps}\cup \bar L_{\eps})\cap L^2(L_{\eps}\cup \bar L_{\eps})\cap L^{\infty}(L_{\eps}\cup \bar L_{\eps})}\le ce^{-\frac{\eps^2\alpha\beta^2}{k_0^2}t},
\ee
\be\label{omd}
||\om^d||_{L^1(L_{0}\cup \bar L_{0})\cap L^2(L_{0}\cup \bar L_{0})\cap L^{\infty}(L_{0}\cup \bar L_{0})}\le ce^{-\frac{\alpha\beta^2}{4k_0^2}t},
\ee
\end{subequations}
Moreover,
\be\label{om'}
||\om'||_{L^2(\Sig^{(3)})}\le \frac{c}{t^{\frac{1}{4}}},\qquad ||\om'||_{L^1(\Sig^{(3)})}\le \frac{c}{t^{\frac{1}{2}}}
\ee
\end{proposition}
\begin{proof}
Consequence of proposition 3.2, and analogous calculations as in lemma 2.13 of \cite{dz}. Let us show equation (\ref{om'}).
\par
From the appendix, we have
\be
|R(k)|\le C(k_0)(1+|k|^5)^{-1}
\ee
on the contour $k=\{k_0+uk_0e^{i\frac{3\pi}{4}},-\infty<u\le \eps\}$, $\eps\le \frac{1}{\sqrt{2}}$.
\par
Since
\be
\ba{rrl}
\re i\tha(k)&=&\frac{\alpha\beta^2}{4k_0^2}\frac{(u^2-\sqrt{2}u)^2(u^2-\sqrt{2}u+2)}{(u^2-\sqrt{2}u+1)^2}\\
&\ge &\frac{\alpha\beta^2}{4k_0^2}Ku^2
\ea
\ee
on the contour $k=\{k_0+uk_0e^{i\frac{3\pi}{4}},-\eps<u\le \eps\}$,
and
\be
\re i\tha(k)\ge -\frac{\alpha\beta^2}{4k_0^2}K'u
\ee
on the contour $k=\{k_0+uk_0e^{i\frac{3\pi}{4}},-\infty<u\le -\eps\}$,
where $K$ and $K'$ are positive constants.
\par
We have the similar estimates on the other parts of the contour $\Sig^{(3)}$.
\par
Moreover,
\be
\ba{rrl}
||\om'||^2_{L^2(\Sig^{(2)})}&=&||\om'||^2_{L^2(\Sig^{(3)})}\\
&\le &C_1(k_0)\int_{\Sig^{(3)}}\left(e^{-u^2K_1t\frac{\alpha\beta^2}{k_0^2}}+e^{-uK_2t\frac{\alpha\beta^2}{k_0^2}}(1+|k|^5)^{-2}|dk|\right)\\
&\le &C_2(k_0)\left(\int_{\R}e^{-u^2K_1t\frac{\alpha\beta^2}{k_0^2}}k_0du+\int_{\R}e^{-uK_2t\frac{\alpha\beta^2}{k_0^2}}k_0du\right)\\
&\le &C_3(k_0)\left(\frac{k^2_0}{\alpha \beta^2t}\right)^{\frac{1}{2}},
\ea
\ee
where $K_1,K_2$ are constants.
\end{proof}
\begin{proposition}
As $t\rightarrow \infty$ and $0<k_0<M$, $||(1-C_{\om})^{-1}||_{L^2(\Sig^{(2)})}\le C$ is equalilent to $||(1-C_{\om'})^{-1}||_{L^2(\Sig^{(2)})}\le C$.
\end{proposition}
\begin{proof}
Consequence of the following inequality, $||C_{\om}-C_{\om'}||_{L^2(\Sig^{(2)})}\le c||\om^e||_{L^2(\Sig^{(2)})}$,
the fact that $||\om^e||_{L^2(\Sig^{(2)})}\le \frac{c}{t^l}$, and the second resolvent identity.
\end{proof}
\begin{proposition}
If $||(1-C_{\om'})^{-1}||_{L^2(\Sig^{(2)})}\le C$, then for arbitrary positive integer $l$,
as $t\rightarrow \infty$ such that $0<k_0<M$,
\be \label{m3sig2om'}
m(x,t)=-\frac{1}{2}([\sig_3,\int_{\Sig^{(2)}}((\id-C_{\om^{'}})^{-1}\id)(\xi)\om^{'}(x,t,\xi)]\frac{d\xi}{2\pi i})_{12}+O(\frac{c}{t^l}).
\ee
\end{proposition}
\begin{proof}
From the second resolvent identity, one can derive the following expression (see equation (2.27) in \cite{dz}),
\be
\ba{rrl}
\int_{\Sig^{(2)}}((1-C_{\om})^{-1}\id)\om\frac{d\xi}{2\pi i}&=&\int_{\Sig^{(2)}}((1-C_{\om'})^{-1}\id)\om'\frac{d\xi}{2\pi i}+\int_{\Sig^{(2)}}\om^e\frac{d\xi}{2\pi i}\\
&&+\int_{\Sig^{(2)}}((1-C_{\om'})^{-1}(C_{\om^e}\id))\om\frac{d\xi}{2\pi i}\\
&&+\int_{\Sig^{(2)}}((1-C_{\om'})^{-1}(C_{\om'}\id))\om^e\frac{d\xi}{2\pi i}\\
&&+\int_{\Sig^{(2)}}((1-C_{\om'})^{-1}C_{\om^e}(1-C_{\om})^{-1})(C_{\om}\id)\om\frac{d\xi}{2\pi i}\\
&=&\int_{\Sig^{(2)}}((1-C_{\om'})^{-1}\id)\om'\frac{d\xi}{2\pi i}+\Rmnum{1}+\Rmnum{2}+\Rmnum{3}+\Rmnum{4}.
\ea
\ee
For $0<k_0<M$, from Proposition (\ref{omfanshu}) it follows that,
\be\label{rmnum1}
\ba{rrl}
|\Rmnum{1}|&\le& ||\om^a||_{L^1(\R\cup i\R)}+||\om^b||_{L^1(L\cup \bar L)}+||\om^c||_{L^1(L_{\eps}\cup \bar L_{\eps})}+||\om^d||_{L^1(L_0\cup \bar L_0)}\\
&\le & ct^{-l},
\ea
\ee
\be\label{rmnum2}
\ba{rrl}
|\Rmnum{2}|&\le& ||(1-C_{\om'})^{-1}||_{L^2(\Sig^{(2)})} ||(C_{\om^e}\id)||_{L^2(\Sig^{(2)})} ||\om||_{L^2(\Sig^{(2)})}\\
&\le & c||\om^e||_{L^2(\Sig^{(2)})} (||\om^e||_{L^2(\Sig^{(2)})}+||\om'||_{L^2(\Sig^{(2)})})\\
&\le & ct^{-l}(ct^{-l}+c)\le ct^{-l},
\ea
\ee
\be\label{rmnum3}
\ba{rrl}
|\Rmnum{3}|&\le & ||(1-C_{\om'})^{-1}||_{L^2(\Sig^{(2)})} ||(C_{\om'}\id)||_{L^2(\Sig^{(2)})} ||\om^e||_{L^2(\Sig^{(2)})}\\
&\le & ct^{-l}
\ea
\ee
\be\label{rmnum4}
\ba{lll}
|\Rmnum{4}|&\le & ||(1-C_{\om'})^{-1}C_{\om^e}(1-C_{\om})^{-1})(C_{\om}\id)||_{L^2(\Sig^{(2)})} ||\om||_{L^2(\Sig^{(2)})}\\
&\le & ||(1-C_{\om'})^{-1}||_{L^2(\Sig^{(2)})} ||C_{\om^e}||_{L^2(\Sig^{(2)})} ||(1-C_{\om})^{-1}||_{L^2(\Sig^{(2)})} ||(C_{\om}\id)||_{L^2(\Sig^{(2)})} ||\om||_{L^2(\Sig^{(2)})}\\
&\le & c ||C_{\om^e}||_{L^2(\Sig^{(2)})} ||(C_{\om}\id)||_{L^2(\Sig^{(2)})} ||\om||_{L^2(\Sig^{(2)})}\\
&\le & c ||\om^e||_{L^2(\Sig^{(2)})} ||\om||^2_{L^2(\Sig^{(2)})}\\
&\le & c t^{-l}.
\ea
\ee
Hence,
\be
|\Rmnum{1}+\Rmnum{2}+\Rmnum{3}+\Rmnum{4}|\le ct^{-l}.
\ee
\par
Applying these estimates to equation (\ref{m2}), we can obtain equation (\ref{m3sig2om'}).
\end{proof}
\par
Let us now show that, in the sense of appropriately defined operator norms, one may
always choose to delete (or add) a portion of a contour(s) on which the jump is $\id$,
without altering the Riemann-Hilbert problem in the operator sense.
\par
Suppose that $\Sig_{1}$ and $\Sig_{2}$ are two oriented skeletons in $\C$ with
\be
\mbox{card}(\Sig_{1}\cap \Sig_{2})<\infty;
\ee
let $u=u(\lam)=u_+(\lam)+u_-(\lam)$ be a $2\times 2$ matrix-valued function on
\be
\Sig_{12}=\Sig_{1}\cup \Sig_{2}
\ee
with entries in $L^2(\Sig_{12})\cap L^{\infty}(\Sig_{12})$ and suppose that
\be
u=0\qquad \qquad \mbox{on }\Sig_{2}.
\ee
Let
\be
R_{\Sig_{1}}\mbox{ denote the restriction map }L^2(\Sig_{12})\rightarrow L^2(\Sig_{1}),
\ee
\be
\id_{\Sig_{1}\rightarrow \Sig^{(12)}}\mbox{ denote the embedding }L^2(\Sig_{1})\rightarrow L^2(\Sig_{12}),
\ee
\be
C_{u}^{12}:L^2(\Sig_{12})\rightarrow L^2(\Sig_{12})\mbox{ denote the operator in (\ref{BCRHPCom}) with }u\leftrightarrow \om,
\ee
\be
C_u^1:L^2(\Sig_{1})\rightarrow L^2(\Sig_{1})\mbox{ denote the operator in (\ref{BCRHPCom}) with }u \uparrow \Sig_{1}\leftrightarrow \om,
\ee
\be
C_u^E:L^2(\Sig_{1})\rightarrow L^2(\Sig_{12})\mbox{ denote the restriction of $C_u^{12}$ to }L^2(\Sig_{1}).
\ee
And, finally, let
\be
\left\{
\ba{l}
\id_{\Sig_{1}}\mbox{ and }\id_{\Sig_{12}}\mbox{ denote the identity operators on}\\
L^2(\Sig_{1})\mbox{ and }L^2(\Sig_{12}),\mbox{ respectively}.
\ea
\right.
\ee
We then have the next lemma:
\begin{lemma} \label{opralema}
\be \label{ideopra1}
C_u^{12}C_u^E=C_u^EC_u^{12},
\ee
\be \label{ideopra2}
(\id_{\Sig_{1}}-C_u^1)^{-1}=R_{\Sig_{1}}(\id_{\Sig_{12}}-C_u^{12})^{-1}\id_{\Sig_{1}\rightarrow \Sig_{12}},
\ee
\be \label{ideopra3}
(\id_{\Sig_{12}}-C_u^{12})^{-1}=\id_{\Sig_{12}}+C_u^E(\id_{\Sig_{1}}-C_u^1)^{-1}R_{\Sig_{1}},
\ee
in the sense that if the right-hand side of (\ref{ideopra2}),resp. (\ref{ideopra3}), exists,
then the left-hand side exists and identity (\ref{ideopra2}),resp. (\ref{ideopra3}), holds true.
\end{lemma}
\begin{proof}
See Lemma 2.56 in \cite{dz}.
\end{proof}
\par
We apply this lemma to the case $u=\om'$, $\Sig_{12}=\Sig^{(2)}$ and $\Sig_{1}=\Sig^{(3)}$. From identity (\ref{ideopra2}),
we get the following proposition, which is the main result of this subsection.
\begin{proposition}\label{m3pro}
\be\label{m3}
m(x,t)=-\frac{1}{2}([\sig_3,\int_{\Sig^{(3)}}(\id-C_{\om'})^{-1}(\xi)\om'(x,t,\xi)]\frac{d\xi}{2\pi i})_{12}.
\ee
\end{proposition}
\par
Set
$$L'=L\backslash L_{\eps}$$.
Then, $\Sig^{(3)}=L'\cup \bar L'$. On $\Sig^{(3)}$, set
$\mu^{'}=(1^{\Sig^{(3)}}-C^{\Sig^{(3)}}_{\om'})^{-1}\id$. Then,
\be
M^{(3)}(x,t,k)=\id+\int_{\Sig^{(3)}}\frac{\mu^{'}(\xi)\om'(\xi)}{\xi-k}\frac{d\xi}{2\pi i}
\ee
solves the Riemann-Hilbert problem
\be\label{M3RHP}
\left\{
\ba{ll}
M^{(3)}_+(x,t,k)=M^{(3)}_-(x,t,k)J^{(3)}(x,t,k),& k\in\Sig^{(3)},\\
M^{(3)}\rightarrow \id,&k\rightarrow \infty.
\ea
\right.
\ee
where
\bea\label{M3canshu}
&\om'=\om'_++\om'_-,&\\
&b'_{\pm}=\id \pm \om'_{\pm},&\\
&J^{(3)}(x,t,k)=(b'_-)^{-1}b'_+&
\eea


\subsubsection{The Scaling operators}
\par
In this subsection, we make a further simplification of the Riemann-Hilbert problem on the truncated contour $\Sig^{(3)}$
by reducing it to the one which is stated on the four disjoint crosses, $\Sig_{A'},\Sig_{B'},\Sig_{C'}$ and $\Sig_{D'}$,
and prove that the leading term of the asymptotic expansion for $m(x,t)$ (proposition \ref{m3pro}, (\ref{m3})) can be written as the sum of four terms
corresponding to the solutions of four auxiliary Riemann-Hilbert problems, each of which is set on one of the crosses;
moreover, the solution of the latter Riemann-Hilbert problem can be presented in terms of an exactly solvable model matrix Riemann-Hilbert
problem, which is studied in the next subsection.
\par
Let us prepare the notations which are needed for exact formulations. Write $\Sig^{(3)}$ as the disjoint union of the four crosses, $\Sig_{A'},\Sig_{B'},\Sig_{C'}$ and $\Sig_{D'}$, extend the contours $\Sig_{A'},\Sig_{B'},\Sig_{C'}$ and $\Sig_{D'}$ (with orientations unchanged) to the following ones,
\[
\ba{c}
\hat \Sig_{A'}=\{k=k_0+uk_0e^{\pm\frac{3i\pi}{4}},u\in \R\},\\
\hat \Sig_{B'}=\{k=-k_0+uk_0e^{\pm\frac{i\pi}{4}},u\in \R\},\\
\hat \Sig_{C'}=\{k=ik_0+uk_0e^{-\frac{i\pi}{4}},u\in \R\}\cup \{k=ik_0+uk_0e^{-\frac{3i\pi}{4}},u\in \R\},\\
\hat \Sig_{D'}=\{k=-ik_0+uk_0e^{\frac{i\pi}{4}},u\in \R\}\cup \{k=-ik_0+uk_0e^{\frac{3i\pi}{4}},u\in \R\}.
\ea
\]
and define by $\Sig_{A},\Sig_{B},\Sig_{C}$ and $\Sig_{D}$, respectively, the contours $\{k=uk_0e^{\pm\frac{i\pi}{4}},u\in \R\}$ oriented
inward as in $\Sig_{A'}$ and $\hat \Sig_{A'}$, inward as in $\Sig_{B'}$ and $\hat \Sig_{B'}$, outward as in $\Sig_{C'}$ and $\hat \Sig_{C'}$, and
outward as in $\Sig_{D'}$ and $\hat \Sig_{D'}$, respectively.
\par

We introduce the scaling operators:
\begin{subequations}
\be\label{scalA}
N_A: f(k)\rightarrow (N_Af)(k)=f(\frac{k_0^2}{2\sqrt{\alpha}\beta \sqrt{t}}k+k_0)
\ee
\be\label{scalB}
N_B: f(k)\rightarrow (N_Bf)(k)=f(\frac{k_0^2}{2\sqrt{\alpha}\beta \sqrt{t}}k-k_0)
\ee
\be\label{scalC}
N_C: f(k)\rightarrow (N_Cf)(k)=f(\frac{-k_0^2}{2\sqrt{\alpha}\beta \sqrt{t}}k+ik_0)
\ee
\be\label{scalD}
N_D: f(k)\rightarrow (N_Df)(k)=f(\frac{-k_0^2}{2\sqrt{\alpha}\beta \sqrt{t}}k-ik_0)
\ee
\end{subequations}
Considering the action of the operators $N_{k},k\in\{A,B,C,D\}$ on $\dta(k) e^{-it\tha(k)}$, we find that,
\be\label{Na}
(N_A\dta e^{-it\tha})(k)=\dta_A^0(k)\dta_A^1(k)
\ee
where
\begin{subequations}
\be\label{scalk0dta0}
\dta_A^0(k)=\frac{k_0^{i\nu-2i\tilde\nu}}{(\sqrt{\alpha t}\beta)^{i\nu}}2^{-i\tilde \nu}e^{i\alpha\beta t-i\frac{\alpha\beta^2}{2k_0^2}t}e^{\chi_{\pm}(k_0)}e^{\tilde \chi_{\pm}'(k_0)}
\ee
\be\label{scalk0dta1}
\ba{rl}
\dta_A^1(k)=&k^{i\nu}e^{-i\frac{k^2}{4}+i\frac{k^6_0 k^3}{\zeta^5\sqrt{t}}}\frac{k_0^{2i\tilde \nu+i\nu}}{2^{i\nu-i\tilde \nu}}\frac{(\frac{k^2_0}{2\sqrt{\alpha t}\beta}k+2k_0)^{i\nu}}{(\frac{k^2_0}{2\sqrt{\alpha t}\beta}k+k_0)^{2i\nu}}\\
&((\frac{k^2_0}{2\sqrt{\alpha t}\beta}k+k_0+ik_0)(\frac{k^2_0}{2\sqrt{\alpha t}\beta}k+k_0-ik_0))^{-i\tilde \nu}\\
&e^{\chi_{\pm}(\frac{k^2_0}{2\sqrt{\alpha t}\beta}k+k_0)-\chi_{\pm} (k_0)}e^{\tilde \chi_{\pm}'(\frac{k^2_0}{2\sqrt{\alpha t}\beta}k+k_0)-\tilde \chi_{\pm}'(k_0)}
\ea
\ee
\end{subequations}
with
\be\label{scaltildechipm}
\tilde \chi_{\pm}'(k)=e^{-\frac{1}{2\pi i}\int_{\pm k_0}^{0}\ln |k-ik'|d\ln(1+|r(ik')|^2)}
\ee
And
\be\label{Nb}
(N_B\dta e^{-it\tha})(k)=\dta_B^0(k)\dta_B^1(k)
\ee
where
\begin{subequations}
\be\label{scal-k0dta0}
\dta_B^0(k)=\frac{k_0^{i\nu-2i\tilde\nu}}{(\sqrt{\alpha t}\beta)^{i\nu}}2^{-i\tilde \nu}e^{i\alpha\beta t-i\frac{\alpha\beta^2}{2k_0^2}t}e^{\chi_{\pm}(-k_0)}e^{\tilde \chi_{\pm}'(-k_0)}
\ee
\be\label{scal-k0dta1}
\ba{rl}
\dta_B^1(k)=&(-k)^{i\nu}e^{-i\frac{k^2}{4}+i\frac{k^6_0 k^3}{\zeta^5\sqrt{t}}}\frac{(-k_0)^{2i\tilde \nu+i\nu}}{2^{i\nu-i\tilde \nu}}\frac{(\frac{k^2_0}{2\sqrt{\alpha t}\beta}k-2k_0)^{i\nu}}{(\frac{k^2_0}{2\sqrt{\alpha t}\beta}k-k_0)^{2i\nu}}\\
&((\frac{k^2_0}{2\sqrt{\alpha t}\beta}k-k_0+ik_0)(\frac{k^2_0}{2\sqrt{\alpha t}\beta}k-k_0-ik_0))^{-i\tilde \nu}\\
&e^{\chi_{\pm}(\frac{k^2_0}{2\sqrt{\alpha t}\beta}k-k_0)-\chi_{\pm} (-k_0)}e^{\tilde \chi_{\pm}'(\frac{k^2_0}{2\sqrt{\alpha t}\beta}k-k_0)-\tilde \chi_{\pm}'(-k_0)}
\ea
\ee
\end{subequations}
with $\tilde \chi_{\pm}'(k)$ defined by (\ref{scaltildechipm}).
\par
For $N_C$,
\be\label{Nc}
(N_C\dta e^{-it\tha})(k)=\dta_C^0(k)\dta_C^1(k)
\ee
where
\begin{subequations}
\be\label{scalik0dta0}
\dta_C^0(k)=\frac{k_0^{2i\nu-i\tilde\nu}}{(\sqrt{\alpha t}\beta)^{-i\tilde \nu}}2^{i\nu}e^{i\alpha\beta t+i\frac{\alpha\beta^2}{2k_0^2}t}e^{\chi'_{\pm}(ik_0)}e^{\tilde \chi_{\pm}(ik_0)}
\ee
\be\label{scalik0dta1}
\ba{rl}
\dta_C^1(k)=&(ik)^{-i\tilde \nu}e^{-i\frac{k^2}{4}+i\frac{k^6_0 k^3}{\zeta^5\sqrt{t}}}\frac{(ik_0)^{-i\tilde \nu}(k_0)^{-2i\nu}}{2^{-i\nu-i\tilde \nu}}\frac{(\frac{-k^2_0}{2\sqrt{\alpha t}\beta}k+ik_0)^{2i\tilde\nu}}{(\frac{-k^2_0}{2\sqrt{\alpha t}\beta}k+2ik_0)^{i\tilde \nu}}\\
&((\frac{-k^2_0}{2\sqrt{\alpha t}\beta}k+ik_0+k_0)(k_0-(\frac{k^2_0}{2\sqrt{\alpha t}\beta}k+ik_0)))^{i \nu}\\
&e^{\chi'_{\pm}(\frac{-k^2_0}{2\sqrt{\alpha t}\beta}k+ik_0)-\chi'_{\pm} (ik_0)}e^{\tilde \chi_{\pm}(\frac{-k^2_0}{2\sqrt{\alpha t}\beta}k+ik_0)-\tilde \chi_{\pm}(ik_0)}
\ea
\ee
\end{subequations}
with
\be\label{scalchipm}
\chi_{\pm}'(k)=e^{-\frac{1}{2\pi i}\int_{0}^{\pm k_0}\ln |k-k'|d\ln(1-|r(k')|^2)}
\ee
For $N_D$
\be\label{Nd}
(N_D\dta e^{-it\tha})(k)=\dta_D^0(k)\dta_D^1(k)
\ee
where

\begin{subequations}
\be\label{scal-ik0dta0}
\dta_D^0(k)=\frac{k_0^{2i\nu-i\tilde\nu}}{(\sqrt{\alpha t}\beta)^{-i\tilde \nu}}2^{i\nu}e^{i\alpha\beta t+i\frac{\alpha\beta^2}{2k_0^2}t}e^{\chi'_{\pm}(-ik_0)}e^{\tilde \chi_{\pm}(-ik_0)}
\ee

\be\label{scal-ik0dta1}
\ba{rl}
\dta_D^1(k)=&(-ik)^{-i\tilde \nu}e^{-i\frac{k^2}{4}+i\frac{k^6_0 k^3}{\zeta^5\sqrt{t}}}\frac{(-ik_0)^{-i\tilde \nu}(k_0)^{-2i\nu}}{2^{-i\nu-i\tilde \nu}}\frac{(\frac{-k^2_0}{2\sqrt{\alpha t}\beta}k-ik_0)^{2i\tilde\nu}}{(\frac{-k^2_0}{2\sqrt{\alpha t}\beta}k-2ik_0)^{i\tilde \nu}}\\
&((\frac{-k^2_0}{2\sqrt{\alpha t}\beta}k-ik_0+k_0)(k_0-(\frac{k^2_0}{2\sqrt{\alpha t}\beta}k-ik_0)))^{i \nu}\\
&e^{\chi'_{\pm}(\frac{-k^2_0}{2\sqrt{\alpha t}\beta}k-ik_0)-\chi'_{\pm} (ik_0)}e^{\tilde \chi_{\pm}(\frac{-k^2_0}{2\sqrt{\alpha t}\beta}k-ik_0)-\tilde \chi_{\pm}(ik_0)}
\ea
\ee
\end{subequations}
\par
Set
\be\label{lDta}
\Dta^0_{l}=(\dta_l^0(k))^{\sig_3},\quad l\in\{A,B,C,D\}
\ee
and let $\tilde \Dta^0_{l}$ denote right multiplication by $\Dta^0_{l}$,
\be\label{lDtarightmulty}
\tilde \Dta^0_{l}\phi=\phi\Dta^0_{l}.
\ee
Denote
\be\label{omk'}
\ba{lll}
\om^{l'}=\left\{\ba{ll}\om',&k\in \Sig_{l'}\\0,&k\in \Sig^{(3)}\backslash\Sig_{l'}\ea\right.&and&\hat \om^{l'}=\left\{\ba{ll}\om^{l'},&k\in \hat \Sig_{l'}\\0,&k\in \hat \Sig_{l'}\backslash \Sig_{l'}\ea\right.
\ea
\ee
According to this.
\be\label{om'fac}
\om'=\sum_{l\in\{A,B,C,D\}}\om^{l'},\quad C^{\Sig^{(3)}}_{\om'}=\sum_{l\in\{A,B,C,D\}}C^{\Sig^{(3)}}_{\om^{l'}}=\sum_{l\in\{A,B,C,D\}}C^{\Sig_{l'}}_{\om^{l'}}.
\ee
\begin{proposition}
For $l,\iota=\{A,B,C,D\}$, $l\ne \iota$ we have
\begin{subequations}
\be
||C^{\Sig^{(3)}}_{\om^{l'}}C^{\Sig^{(3)}}_{\om^{\iota'}}||_{L^{2}(\Sig^{(3)})}\le C(k_0)t^{-\frac{1}{2}},
\ee
\be
||C^{\Sig^{(3)}}_{\om^{l'}}C^{\Sig^{(3)}}_{\om^{\iota'}}||_{L^{\infty}(\Sig^{(3)})\rightarrow L^2(\Sig^{(3)})}\le C(k_0)t^{-\frac{3}{4}}.
\ee
\end{subequations}
\end{proposition}
\begin{proof}
Analogous to lemma 3.5 in \cite{dz}.
\end{proof}
Let us prove some technical results concerning the operators $C^{\Sig_{l'}}_{\om^{l'}}$ and $C^{\hat\Sig_{l'}}_{\hat\om^{l'}}$
\begin{proposition}
For $l\in\{A,B,C,D\}$,
\be\label{ChatomrealCom}
C^{\hat\Sig_{l'}}_{\hat\om^{l'}}=(N_{l})^{-1}\tilde{(\Dta^0_l)}^{-1}C^{\Sig_{l'}}_{\om^{l'}}\tilde{(\Dta^0_l)}N_l,
\quad \om^l=(\Dta^0_l)^{-1}(N_l\hat \om^{l'})\Dta^0_l.
\ee
where
\begin{subequations}
\be
\left.C^{\Sig_{l'}}_{\om^{l'}}\right|_{\bar L_{l}}=-C_+(\cdot \left(\ba{cc}0&0\\\ol{R(\ol{(N_l k)})}(\dta_l^1)^{-2}&0\ea\right)),
\ee
\be
\left.C^{\Sig_{l'}}_{\om^{l'}}\right|_{L_{l}}=C_-(\cdot \left(\ba{cc}0&R((N_l k))(\dta^1_l)^2\\0&0\ea\right)).
\ee
\end{subequations}
here
\begin{subequations}
\be
L_{e}=\{k=\frac{2u\sqrt{\alpha t}\beta}{k_0}e^{-\frac{i\pi}{4}},-\eps<u<\infty\},\quad e=A,B,
\ee
\be
L_{n}=\{k=-\frac{2u\sqrt{\alpha t}\beta}{k_0}e^{\frac{i\pi}{4}},-\eps<u<\infty\},\quad n=C,D.
\ee
\end{subequations}
\end{proposition}
\begin{proof}
We consider the case $l=A$, the cases $l=B,l=C$ and $l=D$ follow in an analogous manner. Since from (\ref{scalk0dta0}), $|\dta^0_A|=1$, it follows from the definition of the operator $\tilde \Dta^0_{A}$ in (\ref{lDta}) that $\tilde \Dta^0_{A}$ is a unitary operator. Then the equation (\ref{ChatomrealCom}) is a
simple change-of-variables argument.
\par
We note that
\be
((\Dta^0_A)^{-1}(N_A\hat \om^{A'})\Dta^0_1)(k)=\left(\ba{cc}0&R((N_A k))(\dta^A_l)^2\\0&0\ea\right)
\ee
on $L_{A}$, otherwise $((\Dta^0_A)^{-1}(N_A\hat \om^{l'})\Dta^0_A)(k)=0$.
Similarly,
\be
((\Dta^0_A)^{-1}(N_A\hat \om^{A'})\Dta^0_l)(k)=\left(\ba{cc}0&0\\\ol{R(\ol{(N_A k)})}(\dta_A^1)^{-2}&0\ea\right)
\ee
on $\bar L_{A}$, otherwise $((\Dta^0_A)^{-1}(N_A\hat \om^{A'})\Dta^0_l)(k)=0$.
\end{proof}
\par
From definitions of $R(k)$, we know that (for case $A$)
\begin{subequations}
\be
R(k_0+)=\lim_{\re k>k_0}R(k)=-\ol{r(k_0)},
\ee
\be
R(k_0-)=\lim_{\re k<k_0}R(k)=\frac{\ol{r(k_0)}}{1-|r(k_0)|^2}.
\ee
\end{subequations}
As $t\rightarrow \infty$,
\be
\bar R(\frac{k_0^2}{2\sqrt{\alpha}\beta \sqrt{t}}k+k_0)(\dta_A^1)^{-2}-\bar R(k_0\pm)k^{-2i\nu}e^{i\frac{k^2}{2}}\rightarrow 0.
\ee
We obtain the following estimate on the rate of convergence:
\begin{proposition}\label{Rinfty}
Let $\kappa$ be a fixed small number with $0<\kappa<\frac{1}{2}$. Then, for $k\in \bar L_A$,
\be
\left|\bar R\left(\frac{k_0^2}{2\sqrt{\alpha t}\beta}k+k_0\right)(\dta^1_A(k))^{-2}-\bar R(k_0\pm)k^{-2i\nu}e^{i\frac{k^2}{2}}\right|\le C(k_0)|e^{i\frac{\kappa}{2}k^2}|\left(\frac{\log t}{\sqrt{t}}\right)
\ee
\end{proposition}
\begin{proposition}(see Proposition 6.2 in \cite{avkahv})
For general operators $C_{\om^{l'}}^{\Sig^{'}},l\in\{1,2,\dots,N\}$, if $(1'-C_{\om^{l'}}^{\Sig^{'}})^{-1}$ exist, then
\begin{subequations}
\be\label{generalop}
(1'-\sum_{1\le X\le N}C_{\om^{X'}}^{\Sig^{'}})(1'+\sum_{1\le Y\le N}C_{\om^{Y'}}^{\Sig^{'}}(1'-C_{\om^{X'}}^{\Sig^{'}})^{-1})=1'-\sum_{1\le Y\le N}\sum_{1\le X\le N}(1-\dta_{XY})C_{\om^{Y'}}^{\Sig^{'}}C_{\om^{X'}}^{\Sig^{'}}(1'-C_{\om^{X'}}^{\Sig^{'}})^{-1}
\ee
and
\be
(1'+\sum_{1\le Y\le N}C_{\om^{Y'}}^{\Sig^{'}}(1'-C_{\om^{X'}}^{\Sig^{'}})^{-1})(1'-\sum_{1\le X\le N}C_{\om^{X'}}^{\Sig^{'}})=1'-\sum_{1\le Y\le N}\sum_{1\le X\le N}(1-\dta_{XY})(1'-C_{\om^{Y'}}^{\Sig^{'}})^{-1}C_{\om^{Y'}}^{\Sig^{'}}C_{\om^{X'}}^{\Sig^{'}}
\ee
where $\dta_{XY}$ is the Kronecker delta.
\end{subequations}
\end{proposition}
\begin{proof}
Assumption the existence of general operators $(1'-C_{\om^{l'}}^{\Sig^{'}})^{-1},l\in\{1,2,\dots,N\}$, induction, and a straightforward application of the second resolvent identity.
\end{proof}
\begin{lemma}\label{fourtermsRH}
If, for $l\in\{A,B,C,D\}$, $(1_{\Sig_{l'}}-C_{\om^{l'}}^{{\Sig_{l'}}})^{-1}$ bounded, then as $t\rightarrow \infty$,
\be\label{fourtermsm}
m(x,t)=-\frac{1}{2}\sum_{l\in\{A,B,C,D\}}\left(\int_{\Sig_{l'}}[\sig_3,((1_{\Sig_{l'}}-C_{\om^{l'}}^{{\Sig_{l'}}})^{-1}\id)(\xi)\om^{l'}(\xi)]\frac{d\xi}{2\pi i}\right)_{12}+O(\frac{C}{t}).
\ee
\end{lemma}
\begin{proof}
Analogous to the proof of Lemma 6.2 in \cite{avkahv}.
\end{proof}
\begin{lemma}
For $l\in\{A,B,C,D\}$, $||(1_{\Sig_{l'}}-C_{\om^{l'}}^{{\Sig_{l'}}})^{-1}||_{L^2}\le C$
\end{lemma}
\begin{proof}
Consider the case $l=A$, the case $l=B,C$ and $l=D$ follow in an analogous manner. From Lemma \ref{opralema}, the boundedness of $(1_{\Sig_{A'}}-C_{\om^{A'}}^{{\Sig_{A'}}})^{-1}$ follows from the boundedness of $(1_{\hat \Sig_{A'}}-C_{\om^{A'}}^{\hat\Sig_{A'}})^{-1}$. From formula (\ref{ChatomrealCom}) we have
\be\label{inversehatC}
(1_{\hat \Sig_{A'}}-C_{\om^{A'}}^{\hat\Sig_{A'}})^{-1}=(N_{A})^{-1}\tilde{(\Dta^0_A)}^{-1}(1_{\Sig_{A'}}-C_{\om^{A'}}^{{\Sig_{A'}}})^{-1}\tilde{(\Dta^0_A)}N_A,
\ee
And then the boundedness of $(1_{\hat \Sig_{A'}}-C_{\om^{A'}}^{\hat\Sig_{A'}})^{-1}$ follows from the boundedness of $(1_{\Sig_{A}}-C_{\om^{A}}^{\Sig_{A}})^{-1}$.
\par
Set
\be
\om^{A}=(\Dta^0_A)^{-1}(N_A\hat \om^{A'})\Dta^0_A,
\ee
so that
\be
C^{\Sig_A}_{\om^A}=C_+(\cdot \om_-^A)+C_-(\cdot \om_+^A).
\ee
On $\Sig_A$, we have the diagram in Figure 9.
\begin{figure}[th]
\centering
\includegraphics{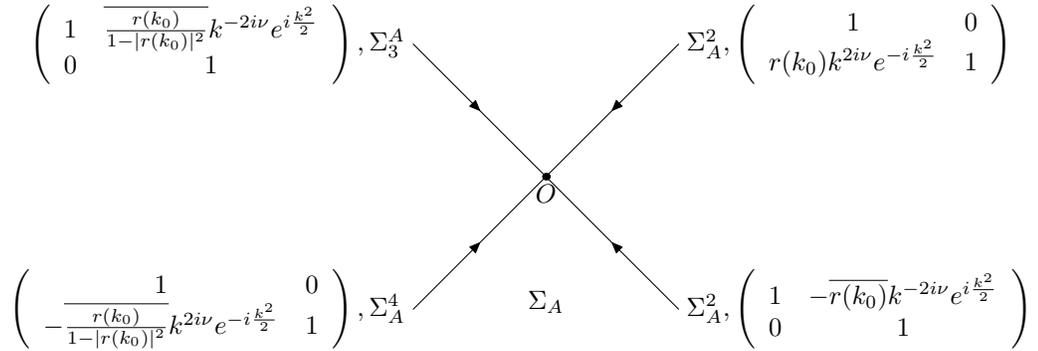}
\caption{The jump condition of cross $k_0$ by scaling.}
\end{figure}
Set $J^{A^0}=(b_-^{A^0})^{-1}b_+^{A^0}=(\id-\om_-^{A^0})^{-1}(\id+\om_+^{A^0})$. Defining as usual $\om^{A^0}=\om_+^{A^0}+\om_-^{A^0}$, and using Proposition \ref{Rinfty}, one finds that
\be
||\om^A-\om^{A^0}||_{L^{\infty}(\Sig_A)\cap L^1(\Sig_A)\cap L^2(\Sig_A)}\le C(k_0)t^{-\frac{1}{2}}.
\ee
Hence, as $t\rightarrow \infty$,
\be
||C^{\Sig_A}_{\om^A}-C^{\Sig_A}_{\om^{A^0}}||_{L^2(\Sig_A)}\le C(k_0)t^{-\frac{1}{2}},
\ee
and consequently, one sees that the boundedness of $(1_{ \Sig_{A}}-C_{\om^{A}}^{\Sig_{A}})^{-1}$ follows from the boundedness of $(1_{ \Sig_{A}}-C_{\om^{A^0}}^{\Sig_{A}})^{-1}$ as $t\rightarrow \infty$.
\par
Then reorient $\Sig_A$ to $\Sig_{A,r}$ as Figure 10.
\begin{figure}[th]
\centering
\includegraphics{LT--FL.10}
\caption{$\Sig_{A,r}$.}
\end{figure}
A simple computation shows that the jump matrix $J^{A,r}=(b_-^{A,r})^{-1}(b_+^{A,r})=(\id-\om_-^{A,r})^{-1}(\id+\om_+^{A,r})$ on $\Sig_{A,r}$ is determined by
\begin{subequations}
\be
\om_{\pm}^{A,r}(k)=-\om_{\mp}^{A^0}(k),\quad for \quad \re k>0,
\ee
and
\be
\om_{\pm}^{A,r}(k)=\om_{\pm}^{A^0}(k),\quad for \quad \re k<0.
\ee
\end{subequations}
\par
The third step is that extending $\Sig_{A,r}\rightarrow \Sig_e=\Sig_{A,r}\cup \R$ with the orientation on $\Sig_{A,r}$ as Figure 10 and the orientation on $\R$ from $-\infty$ to $\infty$. And the jump $J^e=(b_-^e)^{-1}b_=^e=(\id-\om^e_-)^{-1}(\id+\om^e_+)$ with
\begin{subequations}
\be
\om^e(k)=\om^{A,r}(k),\quad k\in \Sig_{A,r},
\ee
\be
\om^e(k)=0,\quad k\in \R.
\ee
\end{subequations}
Set $\C_{\om^e}$ on $\Sig_e$. Once again, by Lemma \ref{opralema}, it is sufficient to bound $(1_{\Sig_e}-C_{\om^e})^{-1}$ on $L^2(\Sig_e)$.
\par
\begin{figure}[th]
\centering
\includegraphics{LT--FL.11}
\caption{$\Sig_{e}$.}
\end{figure}
Then define a piecewise-analytic matrix function $\phi$ as follows:
\[
\tilde M^{(k_0)}=M^{(k_0)}\phi,
\]
where
\[
\phi=\left\{
\ba{ll}
k^{i\nu\sig_3},&k\in \Om^e_2,\Om^e_5,\\
k^{i\nu\sig_3}\left(\begin{array}{cc}1&0\\-r(k_0)e^{-i\frac{k^2}{2}}&1\end{array}\right),&k\in\Om^e_1,\\
k^{i\nu\sig_3}\left(\begin{array}{cc}1&-\overline{r(k_0)}e^{i\frac{k^2}{2}}\\0&1\end{array}\right),&k\in\Om^e_6,\\
k^{i\nu\sig_3}\left(\begin{array}{cc}1&\overline{\frac{r(k_0)}{1-|r(k_0)|^2}}e^{i\frac{k^2}{2}}\\0&1\end{array}\right),&k\in\Om^e_3,\\
k^{i\nu\sig_3}\left(\begin{array}{cc}1&0\\-\overline{\frac{r(k_0)}{1-|r(k_0)|^2}}e^{-i\frac{k^2}{2}}&1\end{array}\right),&k\in\Om^e_4.
\ea
\right.
\]

Thus, we can get the Riemann-Hilbert problem of $\tilde M^{(k_0)}$
\be\label{tildeMk0RHP}
\ba{ll}
\tilde M_+^{(k_0)}(x,t,k)=\tilde M_-^{(k_0)}(x,t,k)J^{e,\phi},\\
\\
\tilde M^{(k_0)}(x,t,k)=(\id+\frac{M^{A^{0}}_1}{k}+O(\frac{1}{k^2}))k^{i\nu\sig_3},&k\rightarrow \infty.
\ea
\ee
where
\be
J^{e,\phi}=\left\{\ba{ll}\left(\ba{cc}1-|r(k_0)|^2&\ol{r(k_0)}e^{-i\frac{k^2}{2}}\\-r(k_0)e^{i\frac{k^2}{2}}&1\ea\right),&k\in\R,\\
\id,&k\in \Sig_{A,r}.
\ea
\right.
\ee
On $\R$ we have
\be
J^{e,\phi}=(b_-^{e,\phi})^{-1}b_+^{e,\phi}=(\id-\om_-^{e,\phi})^{-1}(\id+\om_+^{e,\phi})=
\left(\ba{cc}1&e^{-\frac{ik^2}{2}}\bar r(k_0)\\0&1\ea\right)
\left(\ba{cc}1&0\\-e^{\frac{ik^2}{2}}r(k_0)&1\ea\right).
\ee
Set $C_{e,\phi}=C_{\om^{e,\phi}}=C_+(\cdot \om_-^{e,\phi})+C_-(\cdot \om_+^{e,\phi})$ as thr associated operator on $\Sig_{e}$, with $\om^{e,\phi}=\om_+^{e,\phi}+\om_-^{e,\phi}$. By Lemma \ref{opralema}, the boundedness of $C_{e,\phi}$ follows from the boundedness of the operator
$C_{\om^{e,\phi}|_{\R}}: L^2(\R)\rightarrow L^2(\R)$ associated with the restriction of $\om^{e,\phi}$ to $\R$. Howerover, $||C_{\om^{e,\phi}|_{\R}}||_{L^2(\R)}\le \sup_{k\in \R}{|e^{-\frac{ik^2}{2}}\bar r(k_0)|}\le ||r||_{L^{\infty}(\R)}<1$, and hence, $||(1_{\R}-C_{\om^{e,\phi}|_{\R}})^{-1}||_{L^2(\R)}\le (1-||r||_{L^{\infty}(\R)})^{-1}<\infty$ for all $k_0$, which in turn implies that $(1_{\Sig_e}-C_{e,\phi})^{-1}$ is bounded.
\end{proof}

\subsection{Model Riemann-Hilbert Problem}
In this subsection, we reduce the evaluation of the integrals in Lemma \ref{fourtermsRH} to four Riemann-Hilbert problems on $\R$ which can be solved explicitly.
\par
For $l\in\{A,B,C,D\}$, define
\be
M^l(k)=\id+\int_{\Sig_l}\frac{((1_{\Sig_l}-C_{\om^{l^0}}^{\Sig_l})^{-1}\id)(\xi)\om^{l^0}(\xi)}{\xi-k}\frac{d\xi}{2\pi i},\quad k\in\C\backslash \Sig_l.
\ee
Then, $M^l(k)$ solves the Riemann-Hilbert problem
\be
\left\{
\ba{ll}
M^l_+(k)=M_-^l(k)J^l(k)=M_-^l(k)(\id-\om^{l^0}_-)^{-1}(\id+\om_+^{l^0}),&k\in \Sig_l,\\
M^l(k)\rightarrow \id,&k\rightarrow \infty.
\ea
\right.
\ee
In particular we see that if
\be
M^l(k)=\id+\frac{M_1^l}{k}+O(k^{-2}),\quad k\rightarrow \infty,
\ee
then
\be
M_1^l=-\int_{\Sig_l}((1_{\Sig_l}-C_{\om^{l^0}}^{\Sig_l})^{-1}\id)(\xi)\om^{l^0}(\xi)\frac{d\xi}{2\pi i}.
\ee
Substituting into (\ref{fourtermsm}), we obtain
\be\label{lastm}
m(x,t)=\frac{k_0^2}{2\sqrt{\alpha t}\beta}((\dta^0_A)^{2}(M_1^{A^0})_{12}+(\dta_B^0)^{2}(M_1^{B^{0}})-(\dta_C^0)^2 (M_1^{C^0})-(\dta_D^0)^{2}(M_1^{D^0})_{12})+O(\frac{C}{t}).
\ee
\par
We consider in detail only case $A$.
Write
\be\label{modelk0Psi}
\Psi=\tilde M^{(k_0)}e^{-i\frac{k^2}{4}\sig_3}=\hat \Psi k^{i\nu\sig_3}e^{-i\frac{k^2}{4}\sig_3}
\ee
From formula (\ref{tildeMk0RHP}),
\be\label{modelk0RHP}
\Psi_+(k)=\Psi_-(k)\tilde J^{(k_0)},\quad k\in\R.
\ee
where
\[
J^{(k_0)}=\left(\ba{cc}1-|r(k_0)|^2&\ol{r(k_0)}\\-r(k_0)&1\ea\right)
\]
By differentiation with respect to $k$ and Liouville theorem we can get
\be\label{modelk0}
\frac{d\Psi}{dk}+\frac{1}{2}ik\sig_3\Psi=\beta\Psi,
\ee
where
\[
\beta=\frac{i}{2}[\sig_3,M^{A^{0}}_1]=\left(\ba{cc}o&\beta_{12}\\\beta_{21}&0\ea\right).
\]
Following \cite{dz}(P.350-352), we have
\be\label{k0beta12}
\beta_{12}=-\frac{e^{-\frac{\pi}{2}\nu}}{r(k_0)}\frac{\sqrt{2\pi}e^{i\frac{\pi}{4}}}{\Gam(-i\nu)}.
\ee
Hence,
\be
(M_1^{A^0})_{12}=-i\beta_{12}=i\frac{e^{-\frac{\pi}{2}\nu}}{r(k_0)}\frac{\sqrt{2\pi}e^{i\frac{\pi}{4}}}{\Gam(-i\nu)}.
\ee
From the symmetry reduction for $M(k)$,i.e., $M(-k)=\sig_3M(k)\sig_3$, we have that
\be
(M_1^{A^{0}})_{12}=(M_1^{B^{0}})_{12}.
\ee
For $C$,
\be
\beta_{12}=\frac{e^{\frac{\pi}{2}\tilde \nu}}{r(ik_0)}\frac{\sqrt{2\pi}e^{i\frac{\pi}{4}}}{\Gam(i\tilde \nu)}.
\ee
And similarly $(M_1^{C^{0}})_{12}=(M_1^{D^{0}})_{12}.$
\par
Thus, we have
\begin{theorem}
As $t\rightarrow \infty$, such that $k_0<M$,
\be
\ba{rrl}
m(x,t)&=&\frac{k_0^2}{\beta}\sqrt{\frac{|\nu|}{\alpha t}}e^{i(2\alpha \beta t+2\nu \ln{\frac{k_0}{\sqrt{\alpha t}\beta}}+\frac{\pi}{4}-\frac{\alpha \beta^2}{k_0^2}t-2\tilde \nu \ln{2k_0^2}-2i\chi_{\pm}(k_0)-2i\tilde\chi'_{\pm}(k_0)-\arg{r(k_0)}-\arg{\Gam(-i\nu)})}\\
&&{}-\frac{k_0^2}{\beta}\sqrt{\frac{|\tilde\nu|}{\alpha t}}e^{i(2\alpha \beta t-2\tilde \nu \ln{\frac{k_0}{\sqrt{\alpha t}\beta}}+\frac{\pi}{4}+\frac{\alpha \beta^2}{k_0^2}t+2\nu \ln{2k_0^2}-2i\chi'_{\pm}(ik_0)-2i\tilde\chi_{\pm}(ik_0)-\arg{r(ik_0)}-\arg{\Gam(i\tilde\nu)})}\\
&&{}+O(\frac{1}{t}).
\ea
\ee
\end{theorem}
\begin{proposition}
\be
||2m||^2_{L^2(\R)}=\frac{2}{\pi}\left(\int_{0}^{+\infty}\frac{\log{(1+|r(i\mu)|^2)}}{\mu}d\mu-\int_{0}^{+\infty}\frac{\log{(1-|r(\mu)|^2)}}{\mu}d\mu\right)
\ee
\end{proposition}
\begin{proof}
Analogous to Proposition 8.2 in \cite{avkahv}.
\end{proof}
\begin{lemma}
As $t\rightarrow \infty$,
\be
e^{4i\int_{-\infty}^{x}|m(x;,t)|^2dx'}=e^{\frac{2i}{\pi}\left(\int_{k_0}^{+\infty}\frac{\ln(1+|r(ik')|^2)}{k'}dk'-\int_{k_0}^{+\infty}\frac{\ln(1-|r(k')|^2)}{k'}dk'
-\tilde \psi\right)}+O(\frac{C}{t^{\frac{1}{2}}}).
\ee
where
\be
\tilde \psi=\sqrt{\int_{0}^{k_0}\frac{\ln(1+|r(ik')|^2)}{k'}\frac{\ln(1-|r(k')|^2)}{k'}\cos{(\tilde\xi_1-\tilde\xi_2)dk'}}
\ee
with
$$\tilde\xi_1=2\alpha \beta t+2\nu \ln{\frac{k_0}{\sqrt{\alpha t}\beta}}+\frac{\pi}{4}-\frac{\alpha \beta^2}{k_0^2}t-2\tilde \nu \ln{2k_0^2}-2i\chi_{\pm}(k_0)-2i\tilde\chi'_{\pm}(k_0)-\arg{r(k_0)}-\arg{\Gam(-i\nu)}$$
and
\[
\tilde\xi_2=2\alpha \beta t-2\tilde \nu \ln{\frac{k_0}{\sqrt{\alpha t}\beta}}+\frac{\pi}{4}+\frac{\alpha \beta^2}{k_0^2}t+2\nu \ln{2k_0^2}-2i\chi'_{\pm}(ik_0)-2i\tilde\chi_{\pm}(ik_0)-\arg{r(ik_0)}-\arg{\Gam(i\tilde\nu)}
\]
\end{lemma}
\begin{proof}
Analogous to Lemma 8.1 in \cite{avkahv}.
\end{proof}
\begin{theorem}\label{mainresult}
As $t\rightarrow \infty$,
\be
\ba{rrl}
u_x(x,t)&=&2\frac{k_0^2}{\beta}\sqrt{\frac{|\nu|}{\alpha t}}e^{i(2\alpha \beta t+2\nu \ln{\frac{k_0}{\sqrt{\alpha t}\beta}}+\frac{3\pi}{4}-\frac{\alpha \beta^2}{k_0^2}t-2\tilde \nu \ln{2k_0^2}-2i\chi_{\pm}(k_0)-2i\tilde\chi'_{\pm}(k_0)-\arg{r(k_0)}-\arg{\Gam(-i\nu)})+\tilde\phi}\\
&&{}-2\frac{k_0^2}{\beta}\sqrt{\frac{|\tilde\nu|}{\alpha t}}e^{i(2\alpha \beta t-2\tilde \nu \ln{\frac{k_0}{\sqrt{\alpha t}\beta}}+\frac{3\pi}{4}+\frac{\alpha \beta^2}{k_0^2}t+2\nu \ln{2k_0^2}-2i\chi'_{\pm}(ik_0)-2i\tilde\chi_{\pm}(ik_0)-\arg{r(ik_0)}-\arg{\Gam(i\tilde\nu)})+\tilde\phi}\\
&&{}+O(\frac{1}{t}).
\ea
\ee
where $\tilde\phi=\frac{2}{\pi}\left(\int_{k_0}^{+\infty}\frac{\ln(1+|r(ik')|^2)}{k'}dk'-\int_{k_0}^{+\infty}\frac{\ln(1-|r(k')|^2)}{k'}dk'
-\tilde \psi\right)$
\end{theorem}
Thus, the solution of the Fokas-Lenells equation $u(x,t)$ can be obtained by integration with respect to $x$. This implies that the leading order
asymptotic of the solution to the Fokas-Lenells equation has order $t^{-\frac{1}{2}}$.

\begin{remark}
Although, Fokas-Lenells equation (\ref{FLe1}) is an evolution equation in $u_x$ and that any solution $u(x,t)$ is undetermined up to $u(x,t)\ra u(x,t)+h(t)$ for an arbitrary function $h(t)$, the requirement that $u$ goes to zero as $|x|\ra \infty$ removes this non-uniqueness.
\end{remark}

\begin{remark}
It is not normal that we get the solution $u_x(x,t)$ in terms of the
solution of Rieman-Hilbert problem (\ref{solfromRHP}). And we find
that if we use the asymptotic behavior of the $M(x,t,k)$ as
$k\rightarrow 0$, we can get the solution of $u(x,t)$ from the
$t-$part of Lax pair (\ref{FLelax}). We will use this to deal with
general initial value problem case in another paper \cite{xf}.
\end{remark}
{\bf Acknowledgements}
The work of  Xu was partially supported by Excellent Doctor Research Funding Project of  Fudan University.
The work described in this paper
was supported by grants from the National Science
Foundation of China (Project No.10971031;11271079), Doctoral Programs Foundation of
the Ministry of Education of China, and the Shanghai Shuguang Tracking Project (project 08GG01).

\appendix
\section{\bf Prove Proposition 3.2 and 3.11.}

{\bf Prove Proposition 3.2.}
\par
For the convenience of reader, we show the details of the procedure of the analytic continuation.
\par
{\bf 1. $\frac{k_0}{2}<|k|<k_0,k\in \R$.}
\par
We just consider $\frac{k_0}{2}<k<k_0$, the case for $-k_0<k<-\frac{k_0}{2}$ is similarly.
\par
Set
\be\label{rholessk0}
\rho(k)=\frac{-r(k)}{1-r(k)\ol{r(\bar k)}}=\frac{-r(k)}{1-|r(k)|^2}.
\ee
We split $\rho(k)$ into even and odd parts, $\rho(k)=H_{e}(k^2)+kH_{o}(k^2)$, where $H_{e}(\cdot)$ and $H_{o}(\cdot)$ are of the Schwartz class.
\par
For any positive integer $m$,
\be\label{He}
H_{e}(k^2)=\mu_0^{e}+\mu_1^{e}(k^2-k_0^2)+\cdots+\mu_m^{e}(k^2-k_0^2)^m+\frac{1}{m!}\int_{k_0^2}^{k^2}H_{e}^{(m+1)}(\gam)(k^2-\gam)^md\gam
\ee
and
\be\label{Ho}
H_{o}(k^2)=\mu_0^{o}+\mu_1^{o}(k^2-k_0^2)+\cdots+\mu_m^{o}(k^2-k_0^2)^m+\frac{1}{m!}\int_{k_0^2}^{k^2}H_{o}^{(m+1)}(\gam)(k^2-\gam)^md\gam.
\ee
Set
\be\label{Rm}
R(k)=R_m(k)=\sum_{i=0}^{m}\mu_i^{e}(k^2-k_0^2)^i+k\sum_{i=0}^{m}\mu_i^{o}(k^2-k_0^2)^i.
\ee
Assume $m=4q+1$, where $q$ is a positive integer. Write
\be
\rho(k)=h(k)+R(k),\quad \frac{k_0}{2}<k<k_0,k\in \R.
\ee
Then
\be
\left.\frac{d^j h(k)}{dk^j}\right|_{\pm k_0}=0,\quad 0\le j\le m.
\ee
And we have
\be
h(k)=\frac{(k^2-k_0^2)^{m+1}}{m!}g(k,k_0)
\ee
where
\small{
\be
g(k,k_0)=\left(\int_{0}^{1}H_{e}^{(m+1)}(k_0^2+u(k^2-k_0^2))(1-u)^{m}du
+k\int_{0}^{1}H_{o}^{(m+1)}(k_0^2+u(k^2-k_0^2))(1-u)^{m}du\right)
\ee
}
and
\be
\left|\frac{d^j g(k,k_0)}{dk^j}\right|\le C,\quad \frac{k_0}{2}\le k\le k_0.
\ee
We will split $h$ as $h(k)=h_{\Rmnum{1}}(k)+h_{\Rmnum{2}}(k)$, where $h_{\Rmnum{1}}$ is small and $h_{\Rmnum{2}}$ has an alytic continuation to $\im k>0$. Thus
\be
\rho=h_{\Rmnum{1}}+(h_{\Rmnum{2}}+R).
\ee
\par
Set $p(k)=(k^2-k_0^2)^q$. Recall
\be
\ba{rrl}
\tha(k)&=&k^2(\frac{x}{t}+\alpha)+\frac{\alpha \beta^2}{4k^2}-\alpha\beta\\
&=&\frac{\alpha\beta^2}{4k_0^4}k^2+\frac{\alpha \beta^2}{4k^2}-\alpha\beta.
\ea
\ee

\par
We define
\be
\left\{
\ba{rrll}
\frac{h}{p}(\tha)&=&\frac{h(k(\tha))}{p(k(\tha))},& \tha(k_0)<\tha<\tha(\frac{k_0}{2}),\\
&=&0,& \tha \le \tha(k_0) \quad or \quad \tha \ge \tha(\frac{k_0}{2}).
\ea
\right.
\ee
As $|\tha|\rightarrow \tha(k_0)=\frac{\alpha\beta^2}{2k_0^2}-\alpha\beta$ and $|\tha|>\frac{\alpha\beta^2}{2k_0^2}-\alpha\beta$, we have $\frac{h}{p}(\tha)=O((k^2(\tha)-k_0^2)^{m+1-q})$ and
\be
\frac{d\tha}{dk}=\frac{\alpha\beta^2(k^4-k_0^4)}{2k_0^4k^3}.
\ee
We claim that $\frac{h}{p}\in H^j(-\infty<\tha<\infty)$ for $0\le j\le \frac{3q+2}{2}$. As by Fourier inversion,
\be
\frac{h}{p}(k)=\int_{-\infty}^{\infty}e^{is\tha(k)}\widehat{(\frac{h}{p})}(s)\bar d s,\quad \frac{k_0}{2}<k<k_0,
\ee
where
\be
\widehat{(\frac{h}{p})}(s)=\int_{\tha(k_0)}^{\tha(\frac{k_0}{2})}e^{-is\tha(k)}\frac{h}{p}(\tha(k))\bar d \tha(k),\quad s\in \R.
\ee
where $\bar d s=\frac{ds}{\sqrt{2\pi}}$ and $\bar d \tha(k)=\frac{d\tha(k)}{\sqrt{2\pi}}$.
\par
Thus,
\be
\ba{l}
\int_{\tha(k_0)}^{\tha(\frac{k_0}{2})}\left|\left(\frac{d}{d\tha}\right)^j \frac{h}{p}(\tha(k))\right|^2|\bar d\tha(k)|\\
=\int_{\frac{k_0}{2}}^{k_0}\left|\left(\frac{2k_0^4k^3}{\alpha\beta^2(k^4-k_0^4)}\frac{d}{dk}\right)^j \frac{h}{p}(k)\right|^2|\frac{\alpha\beta^2(k^4-k_0^4)}{2k_0^4k^3}|\bar dk\le C<\infty,
\ea
\ee
for $0<k_0<M,0\le j\le \frac{3q+2}{2}$. Hence,
\be
\int_{-\infty}^{\infty}(1+s^2)^j|\widehat{(h/p)}(s)|^2ds \le C<\infty
\ee
for $0<k_0<M,0\le j\le \frac{3q+2}{2}$.
\par
Split
\be
\ba{rrl}
h(k)&=&p(k)\int_{t}^{\infty}e^{is\tha(k)}\widehat{(h/p)}(s)\bar ds+p(k)\int_{-\infty}^{t}e^{is\tha(k)}\widehat{(h/p)}(s)\bar ds\\
&=&h_{\Rmnum{1}}(k)+h_{\Rmnum{2}}(k).
\ea
\ee
Thus, for $\frac{k_0}{2}<k<k_0\le M$ and any positive integer $n\le \frac{3q+2}{2}$.
\be
\ba{rrl}
|e^{-2it\tha(k)}h_{\Rmnum{1}}(k)|&\le&|p(k)|\int_{t}^{\infty}|\widehat{(h/p)}(s)|\bar ds\\
&\le&|p(k)|(\int_{t}^{\infty}(1+s^2)^{-n}\bar ds)^{\frac{1}{2}}(\int_{t}^{\infty}(1+s^2)^n|\widehat{(h/p)}(s)|^2\bar ds)^{\frac{1}{2}}\\
&\le&\frac{c}{t^{n-\frac{1}{2}}}.
\ea
\ee
Consider the contour $l_1:k(u)=k_0+uk_0e^{i\frac{3\pi}{4}},0\le u\le \frac{1}{\sqrt{2}}$. Since $\re i\tha(k)$ is positive on this contour, $h_{\Rmnum{2}}(k)$ has an analytic continuation to contours $l_1$.
\par
On the contour $l_1$,
\be
\ba{rrl}
|e^{-2it\tha(k)}h_{\Rmnum{2}}(k)|&\le&|k+k_0|^q(k_0u)^qe^{-t\re i\tha(k)}\int_{-\infty}^{t}e^{(s-t)\re i\tha(k)}\widehat{(h/p)}(s)\bar ds\\
&\le &ck_0^{2q}u^{q}e^{-t\re i\tha(k)}(\int_{-\infty}^{t}(1+s^2)^{-1}\bar ds)^{\frac{1}{2}}(\int_{-\infty}^{t}(1+s^2)|\widehat{(h/p)}(s)|^2\bar ds)^{\frac{1}{2}}\\
&\le&ck_0^{2q}u^{q}e^{-t\re i\tha(k)}.
\ea
\ee
Recall $\tha(k)=\frac{\alpha\beta^2}{4k_0^4}k^2+\frac{\alpha \beta^2}{4k^2}-\alpha\beta$, and set $k=k_1+ik_2$, thus
\be
\ba{rrl}
\re i\tha(k)&=&-2\alpha\beta^2 k_1k_2\frac{(k_1^2+k_2^2)^2-k_0^4}{4k_0^4(k_1^2+k_2^2)^2}\\
&=& \frac{\alpha\beta^2}{4k_0^2}\frac{(u^2-\sqrt{2}u)^2(u^2-\sqrt{2}u+2)}{(u^2-\sqrt{2}u+1)^2}\\
&\ge&\frac{\alpha\beta^2 u^2}{2k_0^2},
\ea
\ee
for $0\le u\le \frac{1}{\sqrt{2}}$.
\par
Thus, on the contour $l_1$
\be
\ba{rrl}
|e^{-2it\tha(k)}h_{\Rmnum{2}}(k)|&\le &ck_0^{2q}u^{q}e^{-t\frac{\alpha\beta^2 u^2}{2k_0^2}}\le ck_0^{2q}u^{q}e^{-t\frac{\alpha\beta^2 u^2}{2M^2}}\\
&\le& \frac{c_1}{t^{\frac{q}{2}}},
\ea
\ee
for $k_0<M$.
\par
Fix $\eps$, $0<\eps<\frac{1}{\sqrt{2}}$. If $k(u)$ is on the contour $l_1$ , $\eps<u<\frac{1}{\sqrt{2}}$, then we obtain
\be
|e^{-2it\tha(k)}R(k)|\le ce^{-\frac{\alpha\beta^2 u^2}{k_0^2}t}\le ce^{-\frac{\eps^2\alpha\beta^2}{M^2}t}
\ee

\par
{\bf 2.$0<|k|<\frac{k_0}{2},k\in \R$.}
\par
We consider $0<k<\frac{k_0}{2}$, the case for $-\frac{k_0}{2}<k<0$ is similarly.
\par
Define
\be
\left\{
\ba{rrll}
\rho(\tha)&=&\rho(k(\tha)),&\tha>\tha(\frac{k_0}{2}),\\
&=&0,&\tha\le \tha(\frac{k_0}{2}).
\ea
\right.
\ee
We claim that $\rho(\tha)\in H^j(-\infty<\tha<\infty)$ for any nonnegative integer $j$.
\par
By Fourier inversion,
\be
\rho(\tha(k))=\int_{-\infty}^{\infty}e^{is\tha(k)}\hat \rho(s)\bar ds,\quad 0<k<\frac{k_0}{2},
\ee
where
\be
\hat \rho(s)=\int_{\tha(\frac{k_0}{2})}^{\infty}e^{-is\tha(k)}\rho(\tha(k))\bar d\tha(k).
\ee
Then,
\be
\ba{l}
\int_{\tha(\frac{k_0}{2})}^{\infty}\left|\left(\frac{d}{d\tha}\right)^j \rho(\tha(k))\right|^2|\bar d\tha(k)|\\
=\int_{0}^{\frac{k_0}{2}}\left|\left(\frac{2k_0^4k^3}{\alpha\beta^2(k^4-k_0^4)}\frac{d}{dk}\right)^j \rho(k)\right|^2|\frac{\alpha\beta^2(k^4-k_0^4)}{2k_0^4k^3}|\bar dk\le C<\infty,
\ea
\ee
for any nonnegative integer $j$, $0<k_0<M$, since $r(k)\rightarrow 0$ rapidly, as $k\rightarrow 0$.
\par
Hence
\be
\int_{-\infty}^{\infty}(1+s^2)^j |\hat \rho(s)|^2\bar ds\le C,
\ee
for any nonnegative integer $j$.
\par
Split
\be
\ba{rrl}
\rho(k)&=&\int_{t}^{\infty}e^{is\tha(k)}\hat \rho(s)\bar ds+\int_{-\infty}^{t}e^{is\tha(k)}\hat \rho(s)\bar ds\\
&=&h_{\Rmnum{1}}(k)+h_{\Rmnum{2}}(k).
\ea
\ee
Then, for $0<k<\frac{k_0}{2}$ and any positive integer $j$, we obtain,
\be
\ba{rrl}
|e^{-2it\tha(k)}h_{\Rmnum{1}}(k)|&\le& \int_{t}^{\infty}|\hat \rho|\bar ds\\
&\le& (\int_{t}^{\infty}(1+s^2)^{-j}\bar ds)^{\frac{1}{2}}(\int_{t}^{\infty}(1+s^2)^j|\hat \rho(s)|^2\bar ds)^{\frac{1}{2}}\\
&\le& \frac{c}{t^{j-\frac{1}{2}}}.
\ea
\ee
Consider the contour $l_2: k(u)=uk_0e^{i\frac{\pi}{4}},0<u<\frac{1}{\sqrt{2}}$. Since $\re i\tha(k)$ is positive on this contour, $h_{\Rmnum{2}}$ has an analytic
continuation to contour $l_2$.
\par
On the contour $l_2$,
\be
\ba{rrl}
|e^{-2it\tha(k)}h_{\Rmnum{2}}(k)|&\le&e^{-t\re i\tha(k)}\int_{-\infty}^{t}e^{(s-t)\re i\tha(k)}|\hat \rho(k)|\bar ds\\
&\le&e^{-t\re i\tha(k)}(\int_{-\infty}^{t}(1+s^2)^{-1}\bar ds)^{\frac{1}{2}}(\int_{-\infty}^{t}(1+s^2)|\hat \rho(k)|^2\bar ds)^{\frac{1}{2}},
\ea
\ee
where
\be
\ba{rrl}
\re i\tha(k)&=&-2\alpha\beta^2 k_1k_2\frac{(k_1^2+k_2^2)^2-k_0^4}{4k_0^4(k_1^2+k_2^2)^2}\\
&=&-\frac{\alpha\beta^2}{4k_0^2}\frac{u^4-1}{u^2}\\
&\ge&\frac{\alpha\beta^2}{4k_0^2}
\ea
\ee
for $0<u\le \frac{1}{\sqrt{2}}$.
\par
Thus, we obtain,
\be
|e^{-2it\tha(k)}h_{\Rmnum{2}}(k)|\le ce^{-t\frac{\alpha\beta^2}{4k_0^2}}.
\ee

\par
{\bf 3. $|k|>k_0,k\in \R$}
\par
We consider $k>k_0$, the case for $k<-k_0$ is similarly.
\par
Set
\be
\rho(k)=r(k).
\ee
We write
\be
(k-i)^{m+5}\rho(k)=\mu_0+\mu_1(k-k_0)+\cdots+\mu_m(k-k_0)^m+\frac{1}{m!}\int_{k_0}^{k}((\cdot-i)^{m+5}\rho(\cdot))^{(m+1)}(\gam)(k-\gam)^md\gam.
\ee
Define
\be
R(k)=\frac{\sum_{i=0}^{m}\mu_i(k-k_0)^i}{(k-i)^{m+5}}
\ee
and write $\rho(k)=h(k)+R(k)$. We have
\be
\left.\frac{d^jh(k)}{dk^j}\right|_{k_0}=0,\quad 0\le j\le m.
\ee
For $0<k_0<M$, set
\be
v(k)=\frac{(k-k_0)^q}{(k-i)^{q+2}}.
\ee
Let
\be
\left\{
\ba{rrll}
\frac{h}{v}(\tha)&=&\frac{h}{v}(k(\tha)),&\tha>\tha(k_0),\\
&=&0,&\tha\le \tha(k_0).
\ea
\right.
\ee
Then
\be
\frac{h}{v}(\tha(k))=\int_{-\infty}^{\infty}e^{is\tha(k)}\widehat{(\frac{h}{v})}(s)\bar ds,\quad k\ge k_0,
\ee
where
\be
\widehat{(\frac{h}{v})}(s)=\int_{\tha(k_0)}^{\infty}e^{-is\tha(k)}\frac{h}{v}(\tha(k))\bar d\tha(k).
\ee
Moreover, we have
\be
\frac{h}{v}(\tha(k))=\frac{(k-k_0)^{3q+2}}{(k-i)^{3q+4}}g(k,k_0),
\ee
where
\be
g(k,k_0)=\frac{1}{m!}\int_{0}^{1}((\cdot-i)^{m+5}\rho(\cdot))^{(m+1)}(k_0+u(k-k_0))(1-u)^k du
\ee
and
\be
\left|\frac{d^j g(k,k_0)}{dk^j}\right|\le C,\quad k\ge k_0.
\ee
Using the identity $\left|\frac{k-k_0}{k+k_0}\right|\le 1$ for $k\ge k_0$, we have
\be
\ba{rrl}
\int_{\tha(k_0)}^{\infty}\left|\left(\frac{d}{d\tha}\right)^j\left(\frac{h}{v}(\tha(k))\right)\right|^2\bar d\tha(k)&=&
\int_{k_0}^{\infty}\left|\left(\frac{2k_0^4}{k\alpha\beta^2(1-\frac{k_0^4}{k^4})}\frac{d}{dk}\right)^j\frac{h}{v}(k)\right|^2\frac{k\alpha\beta^2(1-\frac{k_0^4}{k^4})}{2k_0^4}|\bar dk\\
&\le &c\int_{k_0}^{\infty}\left|\frac{(k-k_0)^{3q+2-3j}}{(k-i)^{3q+4}}\right|^2k^{6j-3}(k^4-k_0^4)\bar dk \le C_1, \quad 0\le j\le \frac{3q+2}{3}.
\ea
\ee
Thus,
\be
\int_{-\infty}^{\infty}(1+s^2)^j \left|\widehat{(\frac{h}{v})}(s)\right|^2\bar ds\le C<\infty.
\ee
We write
\be
\ba{rrl}
h(k)&=&v(k)\int_{t}^{\infty}e^{is\tha(k)}\widehat{(h/v)}(s)\bar ds+v(k)\int_{-\infty}^{t}e^{is\tha(k)}\widehat{(h/v)}(s)\bar ds\\
&=&h_{\Rmnum{1}}(k)+h_{\Rmnum{2}}(k).
\ea
\ee
For $k\ge k_0,0<k_0<M$, and any positive integer $e\le \frac{3q+2}{3}$, we obtain,
\be
\ba{rrl}
|e^{-2it\tha(k)}h_{\Rmnum{1}}(k)|&\le& \frac{|k-k_0|^q}{|k-i|^{q+2}}\int_{t}^{\infty}|\widehat{(h/v)}(s)|\bar ds\\
&\le & \frac{|k-k_0|^q}{|k-i|^{q+2}}(\int_{t}^{\infty}(1+s^2)^{-e}\bar ds)^{\frac{1}{2}}(\int_{t}^{\infty}(1+s^2)^e|\widehat{(h/v)}(s)|^2\bar ds)^{\frac{1}{2}}\\
&\le & c\frac{1}{(1+|k|^2)t^{e-\frac{1}{2}}}.
\ea
\ee
And $h_{\Rmnum{2}}(k)$ has an analytic continuation to the lower half-plane, where $\re i\tha(k)$ is positive. We estimate $e^{-2it\tha(k)}h_{\Rmnum{2}}(k)$ on the contour $k(u)=k_0+uk_0e^{-i\frac{\pi}{4}},u\ge 0$.
\par
If $0<u\le M_1$,
\be
|e^{-2it\tha(k)}h_{\Rmnum{2}}(k)|\le c\frac{k_0^qu^qe^{-t\re i\tha(k)}}{|k-i|^{q+2}},
\ee
where
\be
\ba{rrl}
\re i\tha(k)&=&-2\alpha\beta^2 k_1k_2\frac{(k_1^2+k_2^2)^2-k_0^4}{4k_0^4(k_1^2+k_2^2)^2}\\
&=& \frac{\alpha\beta^2}{4k_0^2}\frac{(u^2+\sqrt{2}u)^2(u^2+\sqrt{2}u+2)}{(u^2+\sqrt{2}u+1)^2}\\
&\ge&\frac{\alpha\beta^2 }{4k_0^2}\frac{4u^2}{(u^2+\sqrt{2}u+1)^2}.
\ea
\ee
Then
\be
\ba{rrl}
|e^{-2it\tha(k)}h_{\Rmnum{2}}(k)|&\le &c\frac{k_0^qu^qe^{-t\re i\tha(k)}}{|k-i|^{q+2}}\le c_1\frac{k_0^qu^q}{|k-i|^{q+2}}e^{-t\frac{\alpha\beta^2 }{4k_0^2}\frac{4u^2}{(u^2+\sqrt{2}u+1)^2}}\\
&\le &\frac{c_2}{(1+|k|^2)^{q+2}t^{\frac{q}{2}}}\le \frac{c_2}{(1+|k|^2)t^{\frac{q}{2}}}
\ea
\ee
If $u>M_1$, then
\be
\re i\tha(k)\ge \frac{\alpha\beta^2 }{4k_0^2}\frac{u^6}{(u^2+\sqrt{2}u+1)^2}
\ee
and
\be
\ba{rrl}
|e^{-2it\tha(k)}h_{\Rmnum{2}}(k)|&\le &c\frac{k_0^qu^qe^{-t\re i\tha(k)}}{|k-i|^{q+2}}\le c_3\frac{k_0^qu^q}{|k-i|^{q+2}}e^{-t\frac{\alpha\beta^2 }{4k_0^2}\frac{u^6}{(u^2+\sqrt{2}u+1)^2}}\\
&\le &\frac{c_4}{(1+|k|^2)t^{q}}
\ea
\ee
Hence, for $u>0$, we obtain
\be
|e^{-2it\tha(k)}h_{\Rmnum{2}}(k)|\le \frac{c_5}{(1+|k|^2)t^{\frac{q}{2}}}.
\ee

\par
{\bf 4.$\frac{k_0}{2}<|k|<k_0,k\in i\R$}.
\par
We just consider $\frac{k_0}{2}<\im k<k_0$, the case for $-k_0<\im k<-\frac{k_0}{2}$ is similarly.
\par
Set
\be
\rho(k)=\frac{-r(k)}{1+|r(k)|^2}, \quad \frac{k_0}{2}<\im k<k_0,k\in i\R.
\ee
The following process is similar as the case $\frac{k_0}{2}<k<k_0$. That is,

We split $\rho(k)$ into even and odd parts, $\rho(k)=H_{e}(k^2)+kH_{o}(k^2)$, where $H_{e}(\cdot)$ and $H_{o}(\cdot)$ are of the Schwartz class.
\par
For any positive integer $m$,
\be\label{He}
H_{e}(k^2)=\mu_0^{e}+\mu_1^{e}(k^2+k_0^2)+\cdots+\mu_m^{e}(k^2+k_0^2)^m+\frac{1}{m!}\int_{k_0^2}^{k^2}H_{e}^{(m+1)}(\gam)(k^2-\gam)^md\gam
\ee
and
\be\label{Ho}
H_{o}(k^2)=\mu_0^{o}+\mu_1^{o}(k^2+k_0^2)+\cdots+\mu_m^{o}(k^2+k_0^2)^m+\frac{1}{m!}\int_{k_0^2}^{k^2}H_{o}^{(m+1)}(\gam)(k^2-\gam)^md\gam.
\ee
Set
\be\label{Rm}
R(k)=R_m(k)=\sum_{i=0}^{m}\mu_i^{e}(k^2+k_0^2)^i+k\sum_{i=0}^{m}\mu_i^{o}(k^2+k_0^2)^i.
\ee
Assume $m=4q+1$, where $q$ is a positive integer. Write
\be
\rho(k)=h(k)+R(k),\quad \frac{k_0}{2}<\im k<k_0,k\in i\R.
\ee
Then
\be
\left.\frac{d^j h(k)}{dk^j}\right|_{\pm ik_0}=0,\quad 0\le j\le m.
\ee
And we have
\be
h(k)=\frac{(k^2+k_0^2)^{m+1}}{m!}g(k,k_0)
\ee
where
\small{
\be
g(k,ik_0)=\left(\int_{0}^{1}H_{e}^{(m+1)}(-k_0^2+u(k^2+k_0^2))(1-u)^{m}du
+k\int_{0}^{1}H_{o}^{(m+1)}(-k_0^2+u(k^2+k_0^2))(1-u)^{m}du\right)
\ee
}
and
\be
\left|\frac{d^j g(k,ik_0)}{dk^j}\right|\le C,\quad \frac{k_0}{2}\le \im k\le k_0.
\ee
We will split $h$ as $h(k)=h_{\Rmnum{1}}(k)+h_{\Rmnum{2}}(k)$, where $h_{\Rmnum{1}}$ is small and $h_{\Rmnum{2}}$ has an analytic continuation to $\re k>0$. Thus
\be
\rho=h_{\Rmnum{1}}+(h_{\Rmnum{2}}+R).
\ee
\par
Set $p(k)=(k^2+k_0^2)^q$. Recall
\be
\ba{rrl}
\tha(k)&=&k^2(\frac{x}{t}+\alpha)+\frac{\alpha \beta^2}{4k^2}-\alpha\beta\\
&=&\frac{\alpha\beta^2}{4k_0^4}k^2+\frac{\alpha \beta^2}{4k^2}-\alpha\beta.
\ea
\ee

\par
We define
\be
\left\{
\ba{rrll}
\frac{h}{p}(\tha)&=&\frac{h(k(\tha))}{p(k(\tha))},& \tha(ik_0)<\tha<\tha(\frac{ik_0}{2}),\\
&=&0,& \tha \le \tha(ik_0) \quad or \quad \tha \ge \tha(\frac{ik_0}{2}).
\ea
\right.
\ee
As $|\tha|\rightarrow \tha(ik_0)=-\frac{\alpha\beta^2}{2k_0^2}-\alpha\beta$ and $|\tha|>-\frac{\alpha\beta^2}{2k_0^2}-\alpha\beta$, we have $\frac{h}{p}(\tha)=O((k^2(\tha)+k_0^2)^{m+1-q})$ and
\be
\frac{d\tha}{dk}=\frac{\alpha\beta^2(k^4-k_0^4)}{2k_0^4k^3}.
\ee
We claim that $\frac{h}{p}\in H^j(-\infty<\tha<\infty)$ for $0\le j\le \frac{3q+2}{2}$. As by Fourier inversion,
\be
\frac{h}{p}(k)=\int_{-\infty}^{\infty}e^{is\tha(k)}\widehat{(\frac{h}{p})}(s)\bar d s,\quad \frac{k_0}{2}<\im k<k_0,
\ee
where
\be
\widehat{(\frac{h}{p})}(s)=\int_{\tha(ik_0)}^{\tha(\frac{ik_0}{2})}e^{-is\tha(k)}\frac{h}{p}(\tha(k))\bar d \tha(k),\quad s\in \R.
\ee
where $\bar d s=\frac{ds}{\sqrt{2\pi}}$ and $\bar d \tha(k)=\frac{d\tha(k)}{\sqrt{2\pi}}$.
\par
Thus,
\be
\ba{l}
\int_{\tha(ik_0)}^{\tha(i\frac{k_0}{2})}\left|\left(\frac{d}{d\tha}\right)^j \frac{h}{p}(\tha(k))\right|^2|\bar d\tha(k)|\\
=\int_{i\frac{k_0}{2}}^{ik_0}\left|\left(\frac{2k_0^4k^3}{\alpha\beta^2(k^4-k_0^4)}\frac{d}{dk}\right)^j \frac{h}{p}(k)\right|^2|\frac{\alpha\beta^2(k^4-k_0^4)}{2k_0^4k^3}|\bar dk\le C<\infty,
\ea
\ee
for $0<k_0<M,0\le j\le \frac{3q+2}{2}$. Hence,
\be
\int_{-\infty}^{\infty}(1+s^2)^j|\widehat{(h/p)}(s)|^2ds \le C<\infty
\ee
for $0<k_0<M,0\le j\le \frac{3q+2}{2}$.
\par
Split
\be
\ba{rrl}
h(k)&=&p(k)\int_{t}^{\infty}e^{is\tha(k)}\widehat{(h/p)}(s)\bar ds+p(k)\int_{-\infty}^{t}e^{is\tha(k)}\widehat{(h/p)}(s)\bar ds\\
&=&h_{\Rmnum{1}}(k)+h_{\Rmnum{2}}(k).
\ea
\ee
Thus, for $\frac{k_0}{2}<\im k<k_0\le M$ and any positive integer $n\le \frac{3q+2}{2}$.
\be
\ba{rrl}
|e^{-2it\tha(k)}h_{\Rmnum{1}}(k)|&\le&|p(k)|\int_{t}^{\infty}|\widehat{(h/p)}(s)|\bar ds\\
&\le&|p(k)|(\int_{t}^{\infty}(1+s^2)^{-n}\bar ds)^{\frac{1}{2}}(\int_{t}^{\infty}(1+s^2)^n|\widehat{(h/p)}(s)|^2\bar ds)^{\frac{1}{2}}\\
&\le&\frac{c}{t^{n-\frac{1}{2}}}.
\ea
\ee
Consider the contour $l'_1:k(u)=ik_0+uk_0e^{-i\frac{\pi}{4}},0\le u\le \frac{1}{\sqrt{2}}$. Since $\re i\tha(k)$ is positive on this contour, $h_{\Rmnum{2}}(k)$ has an analytic continuation to contours $l'_1$.
\par
On the contour $l'_1$,
\be
\ba{rrl}
|e^{-2it\tha(k)}h_{\Rmnum{2}}(k)|&\le&|k+ik_0|^q(k_0u)^qe^{-t\re i\tha(k)}\int_{-\infty}^{t}e^{(s-t)\re i\tha(k)}\widehat{(h/p)}(s)\bar ds\\
&\le &ck_0^{2q}u^{q}e^{-t\re i\tha(k)}(\int_{-\infty}^{t}(1+s^2)^{-1}\bar ds)^{\frac{1}{2}}(\int_{-\infty}^{t}(1+s^2)|\widehat{(h/p)}(s)|^2\bar ds)^{\frac{1}{2}}\\
&\le&ck_0^{2q}u^{q}e^{-t\re i\tha(k)}.
\ea
\ee
Recall $\tha(k)=\frac{\alpha\beta^2}{4k_0^4}k^2+\frac{\alpha \beta^2}{4k^2}-\alpha\beta$, and set $k=k_1+ik_2$, thus
\be
\ba{rrl}
\re i\tha(k)&=&-2\alpha\beta^2 k_1k_2\frac{(k_1^2+k_2^2)^2-k_0^4}{4k_0^4(k_1^2+k_2^2)^2}\\
&=& \frac{\alpha\beta^2}{4k_0^2}\frac{(u^2-\sqrt{2}u)^2(u^2-\sqrt{2}u+2)}{(u^2-\sqrt{2}u+1)^2}\\
&\ge&\frac{\alpha\beta^2 u^2}{2k_0^2},
\ea
\ee
for $0\le u\le \frac{1}{\sqrt{2}}$.
\par
Thus, on the contour $l'_1$
\be
\ba{rrl}
|e^{-2it\tha(k)}h_{\Rmnum{2}}(k)|&\le &ck_0^{2q}u^{q}e^{-t\frac{\alpha\beta^2 u^2}{2k_0^2}}\le ck_0^{2q}u^{q}e^{-t\frac{\alpha\beta^2 u^2}{2M^2}}\\
&\le& \frac{c_1}{t^{\frac{q}{2}}},
\ea
\ee
for $k_0<M$.
\par
Fix $\eps$, $0<\eps<\frac{1}{\sqrt{2}}$. If $k(u)$ is on the contour $l'_1$ , $\eps<u<\frac{1}{\sqrt{2}}$, then we obtain
\be
|e^{-2it\tha(k)}R(k)|\le ce^{-\frac{\alpha\beta^2 u^2}{k_0^2}t}\le ce^{-\frac{\eps^2}{M^2}t}
\ee

\par

{\bf 5. $0<|k|<\frac{k_0}{2},k\in i\R$}.
\par
We consider $0<\im k<\frac{k_0}{2},k\in i\R$, the case for $-\frac{k_0}{2}<\im k<0$ is similarly.
\par
Define
\be
\left\{
\ba{rrll}
\rho(\tha)&=&\rho(k(\tha)),&\tha>\tha(i\frac{k_0}{2}),\\
&=&0,&\tha\le \tha(i\frac{k_0}{2}).
\ea
\right.
\ee
We claim that $\rho(\tha)\in H^j(-\infty<\tha<\infty)$ for any nonnegative integer $j$.
\par
By Fourier inversion,
\be
\rho(\tha(k))=\int_{-\infty}^{\infty}e^{is\tha(k)}\hat \rho(s)\bar ds,\quad 0<\im k<\frac{k_0}{2},
\ee
where
\be
\hat \rho(s)=\int_{\tha(i\frac{k_0}{2})}^{\infty}e^{-is\tha(k)}\rho(\tha(k))\bar d\tha(k).
\ee
Then,
\be
\ba{l}
\int_{\tha(i\frac{k_0}{2})}^{\infty}\left|\left(\frac{d}{d\tha}\right)^j \rho(\tha(k))\right|^2|\bar d\tha(k)|\\
=\int_{0}^{i\frac{k_0}{2}}\left|\left(\frac{2k_0^4k^3}{\alpha\beta^2(k^4-k_0^4)}\frac{d}{dk}\right)^j \rho(k)\right|^2|\frac{\alpha\beta^2(k^4-k_0^4)}{2k_0^4k^3}|\bar dk\le C<\infty,
\ea
\ee
for any nonnegative integer $j$, $0<k_0<M$, since $r(k)\rightarrow 0$ rapidly, as $k\rightarrow 0$.
\par
Hence
\be
\int_{-\infty}^{\infty}(1+s^2)^j |\hat \rho(s)|^2\bar ds\le C,
\ee
for any nonnegative integer $j$.
\par
Split
\be
\ba{rrl}
\rho(k)&=&\int_{t}^{\infty}e^{is\tha(k)}\hat \rho(s)\bar ds+\int_{-\infty}^{t}e^{is\tha(k)}\hat \rho(s)\bar ds\\
&=&h_{\Rmnum{1}}(k)+h_{\Rmnum{2}}(k).
\ea
\ee
Then, for $0<\im k<i\frac{k_0}{2}$ and any positive integer $j$, we obtain,
\be
\ba{rrl}
|e^{-2it\tha(k)}h_{\Rmnum{1}}(k)|&\le& \int_{t}^{\infty}|\hat \rho|\bar ds\\
&\le& (\int_{t}^{\infty}(1+s^2)^{-j}\bar ds)^{\frac{1}{2}}(\int_{t}^{\infty}(1+s^2)^j|\hat \rho(s)|^2\bar ds)^{\frac{1}{2}}\\
&\le& \frac{c}{t^{j-\frac{1}{2}}}.
\ea
\ee
Consider the contour $l'_2: k(u)=uk_0e^{i\frac{\pi}{4}},0<u<\frac{1}{\sqrt{2}}$. Since $\re i\tha(k)$ is positive on this contour, $h_{\Rmnum{2}}$ has an analytic continuation to contour $l'_2$.
\par
On the contour $l'_2$,
\be
\ba{rrl}
|e^{-2it\tha(k)}h_{\Rmnum{2}}(k)|&\le&e^{-t\re i\tha(k)}\int_{-\infty}^{t}e^{(s-t)\re i\tha(k)}|\hat \rho(k)|\bar ds\\
&\le&e^{-t\re i\tha(k)}(\int_{-\infty}^{t}(1+s^2)^{-1}\bar ds)^{\frac{1}{2}}(\int_{-\infty}^{t}(1+s^2)|\hat \rho(k)|^2\bar ds)^{\frac{1}{2}},
\ea
\ee
where
\be
\ba{rrl}
\re i\tha(k)&=&-2\alpha\beta^2 k_1k_2\frac{(k_1^2+k_2^2)^2-k_0^4}{4k_0^4(k_1^2+k_2^2)^2}\\
&=&-\frac{\alpha\beta^2}{4k_0^2}\frac{u^4-1}{u^2}\\
&\ge&\frac{\alpha\beta^2}{4k_0^2}
\ea
\ee
for $0<u\le \frac{1}{\sqrt{2}}$.
\par
Thus, we obtain,
\be
|e^{-2it\tha(k)}h_{\Rmnum{2}}(k)|\le ce^{-t\frac{\alpha\beta^2}{4k_0^2}}.
\ee

\par
{\bf 6. $|k|>k_0,k\in i\R$}.
\par
We consider $\im k>k_0,k\in \R$, the case for $\im k<-k_0$ is similarly.
\par
Set
\be
\rho(k)=r(k).
\ee
We write
\be
(k+1)^{m+5}\rho(k)=\mu_0+\mu_1(k-ik_0)+\cdots+\mu_m(k-ik_0)^m+\frac{1}{m!}\int_{k_0}^{k}((\cdot+1)^{m+5}\rho(\cdot))^{(m+1)}(\gam)(k-\gam)^md\gam.
\ee
Define
\be
R(k)=\frac{\sum_{i=0}^{m}\mu_i(k-ik_0)^i}{(k+1)^{m+5}}
\ee
and write $\rho(k)=h(k)+R(k)$. We have
\be
\left.\frac{d^jh(k)}{dk^j}\right|_{ik_0}=0,\quad 0\le j\le m.
\ee
For $0<k_0<M$, set
\be
v(k)=\frac{(k-ik_0)^q}{(k+1)^{q+2}}.
\ee
Let
\be
\left\{
\ba{rrll}
\frac{h}{v}(\tha)&=&\frac{h}{v}(k(\tha)),&\tha>\tha(ik_0),\\
&=&0,&\tha\le \tha(ik_0).
\ea
\right.
\ee
Then
\be
\frac{h}{v}(\tha(k))=\int_{-\infty}^{\infty}e^{is\tha(k)}\widehat{(\frac{h}{v})}(s)\bar ds,\quad k\ge k_0,
\ee
where
\be
\widehat{(\frac{h}{v})}(s)=\int_{\tha(ik_0)}^{\infty}e^{-is\tha(k)}\frac{h}{v}(\tha(k))\bar d\tha(k).
\ee
Moreover, we have
\be
\frac{h}{v}(\tha(k))=\frac{(k-ik_0)^{3q+2}}{(k+1)^{3q+4}}g(k,ik_0),
\ee
where
\be
g(k,ik_0)=\frac{1}{m!}\int_{0}^{1}((\cdot-i)^{m+5}\rho(\cdot))^{(m+1)}(ik_0+u(k-ik_0))(1-u)^k du
\ee
and
\be
\left|\frac{d^j g(k,ik_0)}{dk^j}\right|\le C,\quad \im k\ge k_0.
\ee
Using the identity $\left|\frac{k-ik_0}{k+ik_0}\right|\le 1$ for $\im k\ge k_0$, we have
\be
\ba{rrl}
\int_{\tha(ik_0)}^{\infty}\left|\left(\frac{d}{d\tha}\right)^j\left(\frac{h}{v}(\tha(k))\right)\right|^2\bar d\tha(k)&=&
\int_{ik_0}^{\infty}\left|\left(\frac{2k_0^4}{k\alpha\beta^2(1-\frac{k_0^4}{k^4})}\frac{d}{dk}\right)^j\frac{h}{v}(k)\right|^2\frac{k\alpha\beta^2(1-\frac{k_0^4}{k^4})}{2k_0^4}|\bar dk\\
&\le &c\int_{ik_0}^{\infty}\left|\frac{(k-ik_0)^{3q+2-3j}}{(k+1)^{3q+4}}\right|^2k^{6j-3}(k^4-k_0^4)\bar dk \le C_1, \quad 0\le j\le \frac{3q+2}{3}.
\ea
\ee
Thus,
\be
\int_{-\infty}^{\infty}(1+s^2)^j \left|\widehat{(\frac{h}{v})}(s)\right|^2\bar ds\le C<\infty.
\ee
We write
\be
\ba{rrl}
h(k)&=&v(k)\int_{t}^{\infty}e^{is\tha(k)}\widehat{(h/v)}(s)\bar ds+v(k)\int_{-\infty}^{t}e^{is\tha(k)}\widehat{(h/v)}(s)\bar ds\\
&=&h_{\Rmnum{1}}(k)+h_{\Rmnum{2}}(k).
\ea
\ee
For $\im k\ge k_0,k\in i\R,0<k_0<M$, and any positive integer $e\le \frac{3q+2}{3}$, we obtain,
\be
\ba{rrl}
|e^{-2it\tha(k)}h_{\Rmnum{1}}(k)|&\le& \frac{|k-ik_0|^q}{|k+1|^{q+2}}\int_{t}^{\infty}|\widehat{(h/v)}(s)|\bar ds\\
&\le & \frac{|k-ik_0|^q}{|k+1|^{q+2}}(\int_{t}^{\infty}(1+s^2)^{-e}\bar ds)^{\frac{1}{2}}(\int_{t}^{\infty}(1+s^2)^e|\widehat{(h/v)}(s)|^2\bar ds)^{\frac{1}{2}}\\
&\le & c\frac{1}{(1+|k|^2)t^{e-\frac{1}{2}}}.
\ea
\ee
And $h_{\Rmnum{2}}(k)$ has an analytic continuation to the left half-plane, where $\re i\tha(k)$ is positive. We estimate $e^{-2it\tha(k)}h_{\Rmnum{2}}(k)$ on the contour $k(u)=ik_0+uk_0e^{i\frac{3\pi}{4}},u\ge 0$.
\par
If $0<u\le M_1$,
\be
|e^{-2it\tha(k)}h_{\Rmnum{2}}(k)|\le c\frac{k_0^qu^qe^{-t\re i\tha(k)}}{|k-i|^{q+2}},
\ee
where
\be
\ba{rrl}
\re i\tha(k)&=&-2\alpha\beta^2 k_1k_2\frac{(k_1^2+k_2^2)^2-k_0^4}{4k_0^4(k_1^2+k_2^2)^2}\\
&=& \frac{\alpha\beta^2}{4k_0^2}\frac{(u^2+\sqrt{2}u)^2(u^2+\sqrt{2}u+2)}{(u^2+\sqrt{2}u+1)^2}\\
&\ge&\frac{\alpha\beta^2 }{4k_0^2}\frac{4u^2}{(u^2+\sqrt{2}u+1)^2}.
\ea
\ee
Then
\be
\ba{rrl}
|e^{-2it\tha(k)}h_{\Rmnum{2}}(k)|&\le &c\frac{k_0^qu^qe^{-t\re i\tha(k)}}{|k+1|^{q+2}}\le c_1\frac{k_0^qu^q}{|k+1|^{q+2}}e^{-t\frac{\alpha\beta^2 }{4k_0^2}\frac{4u^2}{(u^2+\sqrt{2}u+1)^2}}\\
&\le &\frac{c_2}{(1+|k|^2)^{q+2}t^{\frac{q}{2}}}\le \frac{c_2}{(1+|k|^2)t^{\frac{q}{2}}}
\ea
\ee
If $u>M_1$, then
\be
\re i\tha(k)\ge \frac{\alpha\beta^2 }{4k_0^2}\frac{u^6}{(u^2+\sqrt{2}u+1)^2}
\ee
and
\be
\ba{rrl}
|e^{-2it\tha(k)}h_{\Rmnum{2}}(k)|&\le &c\frac{k_0^qu^qe^{-t\re i\tha(k)}}{|k+1|^{q+2}}\le c_3\frac{k_0^qu^q}{|k-i|^{q+2}}e^{-t\frac{\alpha\beta^2 }{4k_0^2}\frac{u^6}{(u^2+\sqrt{2}u+1)^2}}\\
&\le &\frac{c_4}{(1+|k|^2)t^{q}}
\ea
\ee
Hence, for $u>0$, we obtain
\be
|e^{-2it\tha(k)}h_{\Rmnum{2}}(k)|\le \frac{c_5}{(1+|k|^2)t^{\frac{q}{2}}}.
\ee

\par
\par
Note that if $l$ is an arbitrary positive integer, we can choose $m$ large enough such that $\frac{3q+2}{2}-\frac{1}{2}>q-\frac{1}{2}>\frac{q}{2}>l$ and Proposition 3.2 holds.

\par
{\bf Prove Proposition 3.11.}
\begin{proof}
Write
\be
 \ba{l}
\bar R\left(\frac{k_0^2}{2\sqrt{\alpha
t}\beta}k+k_0\right)(\dta^1_A(k))^{-2}-\bar
R(k_0\pm)k^{-2i\nu}e^{i\frac{k^2}{2}}
=\\
e^{i\frac{\kappa}{2}k^2}e^{i\frac{\kappa}{2}k^2}\bar
R\left(\frac{k_0^2}{2\sqrt{\alpha
t}\beta}k+k_0\right)k^{-2i\nu}e^{i(1-2\kappa)\frac{k^2}{2}(1-\frac{4k_0^6k}{(1-2\kappa)2\sqrt{\alpha
t}\beta \eta^5})}\\
{}\frac{k_0^{-4i\tilde \nu-2i\nu}}{2^{-2i\nu+2i\tilde
\nu}}\frac{(\frac{k^2_0}{2\sqrt{\alpha
t}\beta}k+2k_0)^{-2i\nu}}{(\frac{k^2_0}{2\sqrt{\alpha
t}\beta}k+k_0)^{-4i\nu}}
 ((\frac{k^2_0}{2\sqrt{\alpha
t}\beta}k+k_0+ik_0)(\frac{k^2_0}{2\sqrt{\alpha
t}\beta}k+k_0-ik_0))^{2i\tilde
\nu}\\
{} e^{-2\left(\chi_{\pm}(\frac{k^2_0}{2\sqrt{\alpha
t}\beta}k+k_0)-\chi_{\pm} (k_0)\right)}e^{-2\left(\tilde
\chi_{\pm}'(\frac{k^2_0}{2\sqrt{\alpha t}\beta}k+k_0)
-\tilde\chi_{\pm}'(k_0)\right)}\\
{}-e^{i\frac{\kappa}{2}k^2}e^{i\frac{\kappa}{2}k^2}\bar
R(k_0\pm)k^{-2i\nu}e^{i(1-2\kappa)\frac{k^2}{2}}
 \ea
\ee
 and also divided it into six terms
 \be
   \bar R\left(\frac{k_0^2}{2\sqrt{\alpha t}\beta}k+k_0\right)(\dta^1_A(k))^{-2}-\bar R(k_0\pm)k^{-2i\nu}e^{i\frac{k^2}{2}}
   =e^{i\kappa\frac{k^2}{2}}(\Rmnum{1}+\Rmnum{2}+\Rmnum{3}+\Rmnum{4}+\Rmnum{5}+\Rmnum{6})
 \ee
where
\[
\ba{l}
\Rmnum{1}=e^{i\kappa\frac{k^2}{2}}k^{-2i\nu}[\bar R(\frac{k_0^2}{2\sqrt{\alpha t}\beta}k+k_0)-\bar R(k_0\pm)]\\
\Rmnum{2}=e^{i\kappa\frac{k^2}{2}}k^{-2i\nu}\bar R(\frac{k_0^2}{2\sqrt{\alpha t}\beta}k+k_0)
   \left(e^{i(1-2\kappa)\frac{k^2}{2}(1-\frac{4k_0^6k}{(1-2\kappa)2\sqrt{\alpha t}\beta\eta^5})}-e^{i(1-2\kappa)\frac{k^2}{2}}\right)\\
\Rmnum{3}=e^{i\kappa\frac{k^2}{2}}k^{-2i\nu}\bar R(\frac{k_0^2}{2\sqrt{\alpha t}\beta}k+k_0)
   e^{i(1-2\kappa)\frac{k^2}{2}(1-\frac{4k_0^6k}{(1-2\kappa)2\sqrt{\alpha t}\beta \eta^5})}
   \left(\frac{2^{2i\nu}}{k_0^{2i\nu}}\frac{(\frac{k_0^2}{2\sqrt{\alpha t}\beta}k+k_0)^{4i\nu}}{(\frac{k_0^2}{2\sqrt{\alpha t}\beta}k+2k_0)^{2i\nu}}-1\right)\\
\Rmnum{4}=e^{i\kappa\frac{k^2}{2}}k^{-2i\nu}\bar R(\frac{k_0^2}{2\sqrt{\alpha t}\beta}k+k_0)
   e^{i(1-2\kappa)\frac{k^2}{2}(1-\frac{4k_0^6k}{(1-2\kappa)2\sqrt{\alpha t}\beta \eta^5})}\\
   {}\frac{2^{2i\nu}}{k_0^{2i\nu}}\frac{(\frac{k_0^2}{2\sqrt{\alpha t}\beta}k+k_0)^{4i\nu}}{(\frac{k_0^2}{2\sqrt{\alpha t}\beta}k+2k_0)^{2i\nu}}
   \left(\frac{((\frac{k_0^2}{2\sqrt{\alpha t}\beta}k+k_0)^2+k_0^2)^{2i\tilde \nu}}{2^{2i\tilde\nu}k_0^{4i\tilde\nu}}-1\right)\\
\Rmnum{5}=e^{i\kappa\frac{k^2}{2}}k^{-2i\nu}\bar R(\frac{k_0^2}{2\sqrt{\alpha t}\beta}k+k_0)
   e^{i(1-2\kappa)\frac{k^2}{2}(1-\frac{4k_0^6k}{(1-2\kappa)2\sqrt{\alpha t}\beta\eta^5})}
   \frac{2^{2i\nu}}{k_0^{2i\nu}}\frac{(\frac{k_0^2}{2\sqrt{\alpha t}\beta}k+k_0)^{4i\nu}}{(\frac{k_0^2}{2\sqrt{\alpha t}\beta}k+2k_0)^{2i\nu}}\\
   {}\frac{((\frac{k_0^2}{2\sqrt{\alpha t}\beta}k+k_0)^2+k_0^2)^{2i\tilde \nu}}{2^{2i\tilde\nu}k_0^{4i\tilde\nu}}
   \left(e^{-2\left(\chi_{\pm}(\frac{k^2_0}{2\sqrt{\alpha t}\beta}k+k_0)-\chi_{\pm} (k_0)\right)}-1\right)\\
\Rmnum{6}=e^{i\kappa\frac{k^2}{2}}k^{-2i\nu}\bar R(\frac{k_0^2}{2\sqrt{\alpha t}\beta}k+k_0)
   e^{i(1-2\kappa)\frac{k^2}{2}(1-\frac{4k_0^6k}{(1-2\kappa)2\sqrt{\alpha t}\beta \eta^5})}
   \frac{2^{2i\nu}}{k_0^{2i\nu}}\frac{(\frac{k_0^2}{2\sqrt{\alpha t}\beta}k+k_0)^{4i\nu}}{(\frac{k_0^2}{2\sqrt{\alpha t}\beta}k+2k_0)^{2i\nu}}\\
   {}\frac{((\frac{k_0^2}{2\sqrt{\alpha t}\beta}k+k_0)^2+k_0^2)^{2i\tilde \nu}}{2^{2i\tilde\nu}k_0^{4i\tilde\nu}}
   e^{-2\left(\chi_{\pm}(\frac{k^2_0}{2\sqrt{\alpha t}\beta}k+k_0)-\chi_{\pm} (k_0)\right)}
   \left(e^{-2\left(\tilde\chi_{\pm}'(\frac{k^2_0}{2\sqrt{\alpha t}\beta}k+k_0)-\tilde\chi_{\pm}'(k_0)\right)}-1\right)\\
\ea
\]
Note that $|e^{i\kappa\frac{k^2}{2}}|=e^{-\kappa u^2\frac{2\alpha t \beta^2}{k_0^2}}$. The terms $\Rmnum{1},\Rmnum{2},\Rmnum{3},\Rmnum{4},\Rmnum{5}$ and $\Rmnum{6}$ can be estimated as follows.
\[
\ba{rrl}
|\Rmnum{1}|&\le &|k^{-2i\nu}||e^{i\kappa\frac{k^2}{2}}||\frac{k_0^2}{2\sqrt{\alpha t}\beta}k|||\partial_{k}\bar R(k)||_{L^{\infty}(\bar L_A)}\\
&\le & \frac{C}{\sqrt{t}},
\ea
\]
where $C$ is independent of $k$.
\[
\ba{rrl}
|\Rmnum{2}|&\le &|k^{-2i\nu}||e^{i\kappa\frac{k^2}{2}}|||\bar R||_{L^{\infty}(\bar L_A)}|\frac{d}{ds}e^{i(1-2\kappa)\frac{k^2}{2}(1-s\frac{4k_0^6k}{(1-2\kappa)2\sqrt{\alpha t}\beta\eta^5})}|,\quad 0<s<1\\
&\le &\frac{C}{\sqrt{t}}
\ea
\]
To estimate $\Rmnum{3}$, we write
\[
\ba{rrl}
|\Rmnum{3}|&\le & |k^{-2i\nu}||e^{i\kappa\frac{k^2}{2}}|||\bar R||_{L^{\infty}(\bar L_A)}|e^{i(1-2\kappa)\frac{k^2}{2}(1-\frac{4k_0^6k}{(1-2\kappa)2\sqrt{\alpha t}\beta \eta^5})}|(\Rmnum{3}_1+\Rmnum{3}_2)\\
\ea
\]
where
\[
\ba{l}
\Rmnum{3}_1=\frac{2^{2i\nu}}{k_0^{2i\nu}}\frac{(\frac{k_0^2}{2\sqrt{\alpha t}\beta}k+k_0)^{4i\nu}}{(\frac{k_0^2}{2\sqrt{\alpha t}\beta}k+2k_0)^{2i\nu}}\left[1-\frac{(\frac{k_0^2}{2\sqrt{\alpha t}\beta}k+2k_0)^{2i\nu}}{(2k_0)^{2i\nu}}\right]\\
\Rmnum{3}_2=\left(\frac{\frac{k_0^2}{2\sqrt{\alpha t}\beta}k+k_0}{k_0}\right)^{4i\nu}-1
\ea
\]
The estimate of $\Rmnum{3}_2$ is as follows,
\[
\ba{rrl}
|\Rmnum{3}_2|&=& |\int_{1}^{1+\frac{k_0}{2\sqrt{\alpha t}\beta}k}4i\nu \xi^{4i\nu-1}d\xi|\\
&\le & \frac{C}{\sqrt{t}},
\ea
\]
as $|\xi^{4i\nu-1}|\le c e^{-4\nu arg \xi}\le c$ for $\xi=1+sk\frac{k_0}{2\sqrt{\alpha t}\beta}=1+sue^{i\frac{\pi}{4}},0\le s\le 1,-\eps <u<\infty$.
Since
the first term on the right-hand side of the equation for $\Rmnum{3}_1$ is bounded, namely,
\[
\left|\frac{2^{2i\nu}}{k_0^{2i\nu}}\frac{(\frac{k_0^2}{2\sqrt{\alpha t}\beta}k+k_0)^{4i\nu}}{(\frac{k_0^2}{2\sqrt{\alpha t}\beta}k+2k_0)^{2i\nu}}\right|\le e^{\frac{\pi \nu}{2}}
\]
one obtains an analogous estimate for $\Rmnum{3}_1$. And the estimate for $\Rmnum{4}$ is similar as $\Rmnum{3}$.
\[
\ba{rrl}
|\Rmnum{5}|&\le & C\sup_{0\le s\le 1}|e^{-2se^{-2\left(\chi_{\pm}(\frac{k^2_0}{2\sqrt{\alpha t}\beta}k+k_0)-\chi_{\pm} (k_0)\right)}}|\left|2e^{i\kappa\frac{k^2}{2}}(\chi_{\pm}(\frac{k^2_0}{2\sqrt{\alpha t}\beta}k+k_0)-\chi_{\pm}(k_0))\right|
\ea
\]
using the Lipschitz property of the function $\log \left(\frac{1-|r(\xi)|^2}{1-|r(k_0)|^2}\right),|\xi|\le k_0$, integrating by parts shows that
\[
\ba{rrl}
\left|2e^{i\kappa\frac{k^2}{2}}(\chi_{\pm}(\frac{k^2_0}{2\sqrt{\alpha t}\beta}k+k_0)-\chi_{\pm}(k_0))\right|&\le &C\frac{\log t}{\sqrt{t}},
\ea
\]
The analogous estimates for $\Rmnum{6}$ can be also obtained.
\end{proof}

\end{document}